\documentclass[aps,pra,reprint,amsmath,superscriptaddress,amssymb,showpacs,nofootinbib]{revtex4-2}
\usepackage{amsmath,braket,amssymb}
\usepackage[final]{graphicx}
\usepackage{subfigure}
\usepackage[english]{babel}
\usepackage[utf8]{inputenc}
\usepackage{mathtools}
\usepackage{empheq}
\usepackage{tensor}
\usepackage{cancel}
\usepackage{enumitem}
\usepackage{simplewick}

\usepackage{enumitem}
\usepackage{slashed}
\usepackage{graphicx}
\usepackage{xcolor}

\colorlet{Green}{black!30!green}

\definecolor{THc}{rgb}{0.9,0.3,0.2}

\usepackage{tikz}
\usetikzlibrary{calc,fadings,decorations.pathreplacing,shapes,shapes.multipart,arrows,shapes.misc,intersections,positioning,patterns}
\tikzset{arrow data/.style 2 args={%
		decoration={%
			markings,
			mark=at position #1 with \arrow{#2}},
		postaction=decorate}
}
\usepackage{bm}
\usepackage[colorlinks=true,citecolor=blue,linkcolor=blue,urlcolor=blue]{hyperref}
\usepackage{dsfont}
\usepackage{changepage}
\usepackage{array}
\usepackage{soul}
\usepackage[capitalise]{cleveref}
\crefname{section}{Sec.}{Secs.}
\Crefname{section}{Sec.}{Secs.}
\usepackage{xifthen}
\usepackage{xargs}
\usepackage{dynkin-diagrams}

\usepackage{amsthm}
\theoremstyle{definition}

\theoremstyle{plain}

\newtheorem{lem}{Lemma}

\usepackage{bm}

\graphicspath{{./}{./figures/}}

\newcommand{\bit}{\begin{itemize}}
	\newcommand{\eit}{\end{itemize}}

\renewcommand{\>}{\right\rangle}
\newcommand{\<}{\left\langle}
\newcommand{\ba}{\begin{align}}
	\newcommand{\ea}{\end{align}}
\newcommand{\be}{\begin{equation}}
	\newcommand{\ee}{\end{equation}}
\newcommand{\bi}{\begin{itemize}}
	\newcommand{\ei}{\end{itemize}}

\newcommand{\Tr}{\operatorname{Tr}}

\def\abs#1{\left|#1\right|}

\DeclareMathAlphabet{\mymathbb}{U}{BOONDOX-ds}{m}{n}

\renewcommand{\log}{\ln}

\usepackage{physics}

\usepackage{braket}
\usepackage[normalem]{ulem}

\newcommand{\idg}[1]{{\bfseries #1)}}

\newcommand{\subfigimg}[3][,]{%
	\setbox1=\hbox{\includegraphics[#1]{#3}}
	\leavevmode\rlap{\usebox1}
	\rlap{\hspace*{2pt}\raisebox{\dimexpr\ht1-0.5\baselineskip}{{\bfseries \large\textsf{#2}}}}
	\phantom{\usebox1}
}

\begin{document}
	\date{\today}

	\newcommand{\bbra}[1]{\<\< #1 \right|\right.}
	\newcommand{\kket}[1]{\left.\left| #1 \>\>}
	\newcommand{\bbrakket}[1]{\< \Braket{#1} \>}
	\newcommand{\pll}{\parallel}
	\newcommand{\nn}{\nonumber}
	\newcommand{\transp}{\text{transp.}}
	\newcommand{\nor}{z_{J,H}}
	
	\newcommand{\hL}{\hat{L}}
	\newcommand{\hR}{\hat{R}}
	\newcommand{\hQ}{\hat{Q}}

	\title{A nonstabilizerness monotone from stabilizerness asymmetry}

\begin{abstract}
We introduce a nonstabilizerness monotone which we name basis-minimised stabilizerness asymmetry (BMSA). It is based on the notion of $G$-asymmetry, a measure of how much a certain state deviates from being symmetric with respect to a symmetry group $G$. For pure states, we show that the BMSA is a strong monotone for magic-state resource theory, while it can be extended to mixed states via the convex roof construction. We discuss its relation with other magic monotones, first showing that the BMSA coincides with the recently introduced basis-minimized measurement entropy, thereby establishing the strong monotonicity of the latter. Next, we provide inequalities between the BMSA and other nonstabilizerness measures known in the literature. Finally, we present numerical methods to compute the BMSA, highlighting its advantages and drawbacks compared to other nonstabilizerness measures in the context of pure many-body quantum states. We also discuss the importance of additivity and strong monotonicity for measures of nonstabilizerness in many-body physics, motivating the search for additional computable nonstabilizerness monotones.
\end{abstract}

\author{Poetri Sonya Tarabunga}
\affiliation{Technical University of Munich, TUM School of Natural Sciences, Physics Department, 85748 Garching, Germany}
\affiliation{Munich Center for Quantum Science and Technology (MCQST), Schellingstr. 4, 80799 M{\"u}nchen, Germany}
\affiliation{The Abdus Salam International Centre for Theoretical Physics (ICTP), Strada Costiera 11, 34151 Trieste, Italy}

\author{Martina Frau}
\affiliation{International School for Advanced Studies (SISSA), Via Bonomea 265, I-34136 Trieste, Italy}
\author{Tobias Haug}
\affiliation{Quantum Research Center, Technology Innovation Institute, Abu Dhabi, UAE}
\author{Emanuele Tirrito}
\affiliation{The Abdus Salam International Centre for Theoretical Physics (ICTP), Strada Costiera 11, 34151 Trieste, Italy}
\affiliation{Dipartimento di Fisica ``E. Pancini", Universit\`a di Napoli ``Federico II'', Monte S. Angelo, 80126 Napoli, Italy}
\author{Lorenzo Piroli}
\affiliation{Dipartimento di Fisica e Astronomia, Universit\`a di Bologna and INFN, Sezione di Bologna, via Irnerio 46, I-40126 Bologna, Italy}

\maketitle
	
	
	\section{Introduction}
 \label{sec:intro}

The theory of nonstabilizerness~\cite{bravyi2005universal,veitch2014resource}, also known as magic, is a well-established branch of quantum information theory. It is based on the notion of stabilizer states~\cite{gottesman1997stabilizer}, which are those generated by the class of Clifford unitary operators and which play a crucial role in quantum computation~\cite{nielsen2011quantum}. Stabilizer states have very special properties, which lie at the basis of most quantum-error correcting codes available today~\cite{kitaev2003fault,eastin2009restriction}.

Roughly speaking, nonstabilizerness is the degree to which a certain state or computation cannot be approximated by a stabilizer state or a Clifford unitary operator, respectively. As Clifford gates can be simulated efficiently by a classical computer~\cite{gottesman1998theory,gottesman1998heisenberg,aaronson2004improved}, a non-zero amount of nonstabilizerness is necessary for universal quantum computation~\cite{kitaev2003fault,eastin2009restriction}. This observation motivates the problem of quantifying nonstabilizerness of quantum states and operators, which is formalized within the framework of the resource theory of nonstabilizerness~\cite{veitch2014resource,chitambar2019quantum}.

The notion of nonstabilizerness was introduced almost twenty years ago~\cite{bravyi2005universal}, but the past few years have witnessed an increasing interest in its study in the context of many-body quantum physics~\cite{white2021conformal,ellison2021symmetry,sarkar2020characterization,sewell2022mana,oliviero2022magic,liu2022many,haug2022quantifying}. Arguably, a strong motivation comes from the recent progress in fault-tolerant quantum computing~\cite{bluvstein2024logical,acharya2024quantum}.
An important point is that Clifford operations can be efficiently implemented fault tolerantly~\cite{eastin2009restriction, preskill1998fault,shor1996fault}, often making them less demanding than other operations~\cite{litinski2019magic,orts2023efficient, orts2023efficient}. Therefore, quantifying nonstabilizerness of a physical state becomes relevant for estimating experimental resources needed to prepare quantum states on fault-tolerant quantum computers~\cite{howard2017application}, as well as a benchmark for the performance of the quantum computer~\cite{oliviero2022measuring,haug2022scalable,haug2023efficient,bluvstein2024logical}. From a more theoretical perspective, nonstabilizerness has been argued to provide further understanding on the structure of many-body states~\cite{gu2024doped,haug2023stabilizer,lami2023nonstabilizerness,tarabunga2024nonstabilizerness,oliviero2022magic,odavic2023complexity}, quantum phases of matter~\cite{white2021conformal,ellison2021symmetry,sarkar2020characterization,liu2022many,haug2022quantifying,tarabunga2023many,tarabunga2024critical,falcao2024nonstabilizerness} and pseudorandomness~\cite{gu2023little,haug2023pseudorandom,bansal2024pseudorandomdensitymatrices}. Further, nonstabilizerness characterises chaotic quantum dynamics~\cite{leone2022stabilizer,leone2021quantum,haferkamp2022random,haug2024probing,lopez2024exact,turkeshi2024magic,dowling2024magic}, the structure of entanglement~\cite{tirrito2023quantifying,gu2024magic,fux2023entanglement,frau2024nonstabilizerness,bejan2024dynamical,tarabunga2024magictransition} and quantum chaos~\cite{lami2024quantum,turkeshi2023measuring,leone2021quantum,leone2023nonstabilizerness,leone2023phase,garcia2023resource}.

Within the standard resource theory~\cite{chitambar2019quantum}, a proper measure of nonstabilizerness should satisfy a minimal set of conditions which define \emph{monotones} with respect to Clifford operations. In addition, from the point of view of many-body physics, it is also very natural to require additional properties such as \emph{additivity} (at least asymptotically in the system size) and \emph{strong monotonicity}~\cite{chitambar2019quantum}, which guarantees that nonstabilizerness does not increase, on average, under local measurements.

Unfortunately, most of the measures of nonstabilizerness known from the standard resource theory are hard to compute in the many-body setting. Recently, a breakthrough came with the introduction of the so-called stabilizer R\'enyi entropies~\cite{leone2022stabilizer}, which were shown to be monotones for R\'enyi index $\alpha\geq 2$~\cite{haug2023stabilizer, leone2024stabilizer}. Differently from several known monotones, they do not require minimization procedures, making their computation feasible in a variety of physically interesting situations~\cite{haug2022quantifying,tarabunga2023many,lami2023nonstabilizerness,haug2023efficient,tarabunga2024nonstabilizerness,tarabunga2023magic,liu2024nonequilibrium}. Yet, while being additive, the SRE is not a strong monotone~\cite{haug2023stabilizer}. At the same time, the closely related linear stabilizer entropy is a strong monotone~\cite{leone2024stabilizer}, but it is not additive. Therefore, it is still important to search for novel nonstabilizerness monotones which are both additive and computable, at least in restricted classes of states. 

Motivated by this picture, we introduce a measure of nonstabilizerness which we name basis-minimised stabilizerness asymmetry (BMSA). It is based on the notion of $G$-asymmetry~\cite{vaccaro2008tradeoff,gour2009measuring}, a measure of how much a certain state deviates from being symmetric with respect to a symmetry group $G$. For pure states, we show that the BMSA is a strong monotone for magic-state resource theory, which can be extended to mixed states via a standard convex roof construction. Next, we discuss in detail its relation with other known magic monotones. In particular, we first show that the BMSA coincides with a recently introduced basis-minimized measurement entropy~\cite{niroula2023phase}, thereby establishing the strong monotonicity of the latter. We then provide inequalities between the BMSA and other known nonstabilizerness measures, highlighting the advantages and drawbacks of the BMSA. Finally, we discuss computational methods to evaluate the BMSA and present numerical results in explicit cases. 

The rest of this work is organized as follows. We begin in Sec.~\ref{sec:resource_theory}, where we recall the resource theory of nonstabilizerness and survey known measures. Here, we also discuss the importance of additivity and strong monotonicity in the many-body setting. Next, in Sec.~\ref{sec:BMSA} we introduce the BSMA and study its properties, proving in particular its strong monotonicity. Its relation with other monotones is discussed in Sec.~\ref{sec:relations}, while our numerical results are reported in Sec.~\ref{sec:numerical_results}. Finally, Sec.~\ref{sec:outlook} contains our conclusions, while the most technical aspects of our work are consigned to several appendices.

\section{Nonstabilizerness resource theory}
\label{sec:resource_theory}

\subsection{Preliminaries}

We begin by introducing some notation and recalling basic definitions. We consider a set of $N$ qubits associated with the Hilbert space $\mathcal{H}_N=\bigotimes_{j=1}^N\mathcal{H}_j$, where $\mathcal{H}_j\simeq \mathbb{C}^2$. We denote by $\sigma^\alpha_j$ the Pauli matrices acting on qubit $j$, where $\alpha=0,1,2,3$ and $\sigma^0=\openone$. In addition, we denote by $\ket{0}_j$, $\ket{1}_j$ the eigenbasis of the Pauli operator $\sigma^z_j$, and we call $\{\ket{0}_j, \ket{1}_j\}$ the local computational basis. 

Next, we denote by $\tilde{P}_N$ the Pauli group consisting of all $N$-qubit Pauli strings with phases $\pm 1$, $\pm i$. The Clifford group is the set of unitary operators $U$ such that $U WU^\dagger\in \tilde{\mathcal{P}}_N$ for all $W\in \tilde{\mathcal{P}}_N$, while (pure) stabilizer states are the states generated by applying elements of the Clifford group to the reference state $\ket{0}^{\otimes N}$. In the following, we will also denote by $\mathcal{P}_N$ the set of Pauli strings with trivial phase $+1$.

Given a pure stabilizer state $|s\rangle$, we introduce the unsigned stabilizer group $G(|s\rangle)$ which is the group of the Pauli operators that stabilize $|s\rangle$ up to a sign. Namely,
\begin{equation}
G(| s \rangle) = \left\lbrace P \in \mathcal{P}_N : P | s \rangle = \pm | s \rangle   \right\rbrace \, .
\end{equation} 
It follows from the standard theory of stabilizer states that $G(|s\rangle)$ is generated by $N$ mutually commuting Pauli operators~\cite{nielsen2011quantum}. Conversely, it can be seen that any state which is stabilized by $N$ mutually commuting Pauli operators can be prepared by Clifford unitary operators. 

For a given stabilizer group $G$, the set of stabilizer states $\{\ket{s} : G(\ket{s}) = G \}$ forms an orthonormal basis, commonly referred to as stabilizer basis. This basis can be constructed by introducing destabilizers $d_i \in \mathcal{P}_N, i=0,1, \dots, 2^N-1$ \cite{aaronson2004improved} associated to a stabilizer state $\ket{s}$, such that the $2^N$ stabilizer states above can be written as $\{d_i^{s}\ket{s}\}_i$.

\subsection{Stabilizer protocols and magic monotones}

As mentioned, a precise definition of nonstabilizerness can be given within the mathematical framework of quantum resource theory. In the following, we only review the aspects which are directly relevant to our work, referring to the literature for a thorough introduction~\cite{veitch2014resource,chitambar2019quantum}

The logic of a quantum resource theory is to indirectly define a certain quantity (the resource) by specifying the free operations, namely the operations which do not increase that quantity. In the case of nonstabilizerness the choice of free operations is not unique~\cite{liu2022many,heimendahl2022axiomatic}, but a minimal set is given by so-called stabilizer protocols. Stabilizer protocols are quantum channels obtained by a sequence of the following elementary operations:
\begin{enumerate}
	\item Clifford unitary operations;\label{eq:clifford_unitary_condition}
	\item Tensor product with stabilizer states, $\rho\mapsto \rho \otimes \sigma$, where $\sigma$ is a stabilizer state;
	\item Measurements in the computational basis;
	\item Discarding of qubits; 
	\item The above operations conditioned on the outcomes of measurements.\label{eq:conditioned_operations}
\end{enumerate}
It is worth noticing that any circuit consisting of these operations can be efficiently simulated on a classical computer~\cite{gottesman1998theory,gottesman1998heisenberg,aaronson2004improved}. 

Given a function $\mathcal{M}$ defined on the set of states, we say that $\mathcal{M}$ is a good measure of nonstabilizerness if it is a \emph{monotone}. That is, $\mathcal{M}$ does not increase under any stabilizer protocol $\mathcal{E}$, 
\begin{equation}
	\label{eq:monotonicity_condition}
\mathcal{M}\left[\mathcal{E}(\rho)\right]\leq \mathcal{M}(\rho)\,.
\end{equation}
Monotonicity with respect to stabilizer protocols is a natural minimal requirement for a proper measure of nonstabilizerness. However, additional properties can be required, which can be suitable depending on the context. We will require in particular that $\mathcal{M}(\ket{\psi})\geq 0$ and that $\mathcal{M}(\ket{\psi})=0$ if and only if $\ket{\psi}$ is a stabilizer state. In addition, we now discuss additivity and strong monotonicity. 

First, we say that $\mathcal{M}$ is additive if
\begin{equation}\label{eq:monotonicity}
    \mathcal{M}[\rho\otimes \sigma]=\mathcal{M}[\rho]+\mathcal{M}[\sigma]\,.
\end{equation}
It is also useful to introduce a weaker condition, known as weak-additivity, which states
\begin{equation}
    \mathcal{M}[\rho^{\otimes N}]= N\mathcal{M}[\rho]\,.
\end{equation}
Next, we define strong monotones. Let $\Lambda$ be a subset containing $m<N$ qubits, and suppose we perform a computational-basis measurement on each qubit in $\Lambda$. We denote by ${\lambda}=\{\lambda_1,\ldots, \lambda_m\}$ the set of outcomes ($\lambda_j=0,1$), by $\rho_{\bf \lambda}$ the post-measurement state, and by $p_{\lambda}={\rm Tr}\left[ \Pi_\lambda \rho \Pi_\lambda \right]$ the corresponding probability, where 
\begin{equation}
\Pi_\lambda= |\lambda\rangle \langle \lambda| \otimes \openone_{N\setminus m}\,.
\end{equation}
The strong-monotonicity condition then reads
\begin{equation}
	\mathcal{M}(\rho)\geq \sum_{\lambda}p_\lambda \mathcal{M}[\rho_\lambda]\,.
	\label{eq:strong_monotonicity}
\end{equation} 
That is, $\mathcal{M}$ does not increase, on average, under computational-basis measurements. Note that Eq.~\eqref{eq:strong_monotonicity}, together with conditions~\ref{eq:clifford_unitary_condition}--\ref{eq:conditioned_operations}, also implies that $\mathcal{M}$ does not increase, on average, if the computational-basis measurements are followed by conditioned Pauli operations.

\subsection{Nonstabilizerness measures}
\label{sec:nonstabilizerness_measures}

We now survey some known nonstabilizerness measures and review some of their properties. In the following, we will mostly focus on pure states $\ket{\psi}$, as this is the main focus of our work.

First, we introduce the min-relative entropy of magic~\cite{bravyi2019simulation,liu2022many}, defined by
\begin{equation}
	D_\text{min}(\ket{\psi})=-\log\left[F_\text{STAB}(\ket{\psi})\right]\,,
\end{equation}
where
\begin{equation}
F_\text{STAB}(\ket{\psi})=\text{max}_{\ket{s}}\{\langle\psi|s\rangle^2\}
\end{equation}
is the stabilizer fidelity. Here, the maximum is taken over the set of stabilizer states $\ket{s}$. That is, $D_\text{min}$ measures the distance between $\ket{\psi}$ and its nearest stabilizer state. The min-relative entropy is a monotone, although not a strong monotone~\cite{haug2023stabilizer}. In addition, $D_{\rm min}$ is sub-additive and, in fact, weakly additive, so that $D_\text{min}(\ket{\psi}^{\otimes N})=N	D_\text{min}(\ket{\psi})$~\cite{bravyi2019simulation}. The min-relative entropy is closely related to the relative entropy of nonstabilizerness~\cite{veitch2014resource} 
\begin{equation}
\label{eq:relative_entropy_of_magic}
    \mathrm{r}_{\mathcal{M}}(\ket{\psi}) \equiv \min _{\sigma \in \operatorname{STAB}\left(\mathcal{H}_{\mathrm{d}}\right)} S(\ket{\psi}\bra{\psi} \| \sigma)\,,
\end{equation}
which is sub-additive and a strong monotone. Here, we introduced the quantum relative entropy
\begin{equation}
\label{eq:relative_entropy}
    S(\rho \| \sigma)=
    \begin{cases}
    +\infty,& {\rm if } \operatorname{supp}(\rho) \cap \operatorname{ker}(\sigma) \neq\{0\}\,,\\
    \operatorname{Tr} \rho(\log \rho-\log \sigma),  & {\rm otherwise}\,.
    \end{cases}
\end{equation}

Next, we define the so-called log-free robustness of magic~\cite{howard2017application,liu2022many,howard2017application}
\begin{equation}
	\text{LR}(\ket{\psi})=\log\left[\text{min}_x\left\{ \sum_i \vert x_i \vert: \ket{\psi}\bra{\psi}=\sum_i x_i \sigma_i\right\}\right]\,,
\end{equation}
where $\mathcal{S}=\{\sigma_i \}$ is the set of pure $N$-qubit stabilizer states. By taking the exponent, we obtain the robustness of magic (RoM) 
\begin{equation}\label{eq:robustness}
R(\ket{\psi})=e^{\text{LR}(\ket{\psi})}\,.
\end{equation}
The RoM is also a strong monotone and sub-additive.

Finally, we recall the recently introduced stabilizer Rényi entropies (SREs)~\cite{leone2022stabilizer}, which read
\begin{equation}\label{eq:SRE}
	M_{n}(|\psi\rangle)=(1-n)^{-1} \log \sum_{P \in \mathcal{P}_{N}} \frac{\langle\psi|P|\psi\rangle^{2n}}{2^N}\,.
\end{equation}
The SREs are nonstabilizerness monotones for integer $n\geq 2$~\cite{leone2024stabilizer}, they are additive~\cite{leone2022stabilizer} but not strong monotones~\cite{haug2023stabilizer}. They are closely related to the linear
stabilizer entropies
\begin{equation}
 M^{\rm lin}_n= \frac{1}{2^N}  \sum_{P \in \mathcal{P}_{N}} \frac{\langle\psi|P|\psi\rangle^{2n}}{2^N}\,,
\end{equation}
so that 
\begin{equation}
    M_{n}(|\psi\rangle)=\frac{1}{1-n}\log\left(1-M^{\rm lin}_n\right)\,.
\end{equation}
The linear stabilizer entropies are monotones and strong monotones for $n\geq 2$, but they are not additive.

\subsection{On additivity and strong monotonicity}

As mentioned, the SRE entropies provide a very powerful tool to estimate nonstabilizerness, as their evaluation does not involve minimization over large sets of states. In addition, while they are not strong monotones, the closely related linear R\'enyi entropies are, so that one could wonder about the necessity of additional measures of nonstabilizerness.

In this section, we elaborate on the conditions of additivity and strong monotonicity, discussing their importance in the context of entangled states of many qubits. As the linear stabilizer entropies and the SREs do not display the first and second conditions, respectively, we argue that it is still of interest to look for novel computable nonstabilizerness measures that display both. 

The following discussion is naturally connected with several topics in general resource theory, such as state convertibility and the difference between deterministic and stochastic transformations~\cite{chitambar2019quantum}. However, we do not enter into any technical aspect, but content ourselves with presenting a discussion based on intuitive arguments.

We first note that additivity, or at least its weak version, is a very natural property to be required when considering large collections of qubits. Intuitively, if we view nonstabilizerness as a resource to be depleted to prepare a given state, the amount of it encoded in the tensor product $\ket{\psi}^{\otimes N}$ should increase roughly linearly in $N$ (for large $N$), since we may imagine to simultaneously prepare blocks $\ket{\psi}^{\otimes N/k}$, for any given integer $k$. 

The strong-subadditivity is also a natural requirement when dealing with states of interest for many-body physics. Without it,  a local measurement may increase arbitrarily the nonstabilizerness of a state with finite probability. This would give rise to a measure which is not robust with respect to local perturbations, which one can not typically rule out in physical contexts.

To be concrete, consider the family of states~\cite{haug2023stabilizer}  
\begin{equation}\label{eq:psi_eps_state}
	\ket{\psi_N^\varepsilon}=\frac{1}{\sqrt{\mathcal{N}_\varepsilon}}\left[\ket{0}^{\otimes N}+ \varepsilon\ket{\chi}^{\otimes N}\right]\,,
\end{equation}
where $\varepsilon$ is a small parameter, $\ket{\chi}$ is the magic state~\cite{bravyi2005universal}
\begin{align}\label{eq:t_state}
	\ket{\chi}&=e^{-i(\pi / 4)}\cos \beta|0\rangle+ \sin \beta|1\rangle\,,
\end{align}
with $\cos(2\beta)=1/\sqrt{3}$, and 
\begin{equation}\label{eq:norm}
	\mathcal{N}_\varepsilon=1+\varepsilon^2+2\varepsilon \cos( \beta)^N[\cos(N \pi / 4)]\,.
\end{equation}
As $N\to\infty$, one expects that $\ket{\psi_N^\varepsilon}$ should be an expensive state from the resource theory point of view, when $N$ is very large. The reason is that a computational-basis measurement on a single qubit, say the first one, has a finite probability (independent of $N$) to transform the state into 
\begin{equation}
    \ket{\psi_N^\prime}=\ket{1}\otimes \ket{\chi}^{\otimes (N-1)}\,,
\end{equation}
which has an arbitrary amount of resource as $N\to \infty$.

This intuition is guaranteed to be captured by any measure $\mathcal{M}$ of nonstabilizerness which is both additive and a strong monotone. To see this, let us denote by $p_1$ the probability to obtain the outcome $1$ when measuring the first qubit. Using $p_1\simeq \frac{\varepsilon^2}{1+\varepsilon^2}\sin\beta^2$~\cite{haug2023stabilizer}, we have
\begin{align}
    \mathcal{M}(\ket{\psi_N^\varepsilon})&\geq p_1 \mathcal{M}(\ket{\psi_N^\prime})
    \nonumber\\
    &=p_1 \mathcal{M}(\ket{0})+(N-1)p_1 \mathcal{M}(\ket{\chi})\nonumber\\
    &\simeq  (N-1)\frac{\varepsilon^2}{1+\varepsilon^2}\sin\beta^2\mathcal{M}(\ket{\chi})\,,
\end{align}
so that $\mathcal{M}(\ket{\psi_N^\varepsilon})$ grows linearly in $N$.

Conversely, if a monotone $\mathcal{M}$ fails to be either additive or a strong monotone, then it would not necessarily allow one to appreciate that $\ket{\psi_N^\varepsilon}$ is a highly rich state from the point of view of nonstabilizerness. For instance, one can show that the SREs satisfy~\cite{haug2023stabilizer}
\begin{equation}\label{eq:const_bound}
	M_n(\ket{\psi_N^\varepsilon})<c_n(\varepsilon)\,, \qquad (n>1)\,,
\end{equation}
where $c_n(\varepsilon)$ is a constant which depends on $n$ and $\varepsilon$, but not on $N$. Therefore, according the SREs, $\ket{\psi^\varepsilon_N}$ apparently corresponds to a modest amount of nonstabilizerness, contrary to our intuitive expectation. At the same time, the linear R\'enyi entropies also appear not be adequate to capture the nonstabilizerness structure of $\ket{\psi^\varepsilon_N}$, since $M^{\rm lin}_n(\ket{\psi^\varepsilon_N})<1$. Ultimately, this is due to the fact that $M^{\rm lin}_n$ are not additive.

Finally, we mention that weak additivity is generally a very strong requirement. However, given a sub-additive measure $\mathcal{M}$ satisfying $\mathcal{M}(\rho\otimes \sigma)\leq \mathcal{M}(\rho)+\mathcal{M}(\sigma)$, one can always construct a regularized version~\cite{chitambar2019quantum}
\begin{equation}
\mathcal{M}^{\infty}(\rho)=\lim _{N \rightarrow \infty} \frac{1}{N} f\left(\rho^{\otimes N}\right)\,,
\end{equation}
which can be easily seen to be weakly additive. Unfortunately, this construction trivializes for the linear stabilizer entropies. Indeed, since $M^{\rm lin}_n(\ket{\psi})\leq 1$ for any state $\ket{\psi}$, we have
\begin{equation}
\lim _{N \rightarrow \infty} \frac{1}{N} M^{\rm lin}_n\left(\ket{\psi}^{\otimes N}\right)=0\,,
\end{equation}
so that the regularized version of $M^{\rm lin}_n$ is always zero. 

\section{The basis-minimised stabilizerness asymmetry}
\label{sec:BMSA}

In this section we introduce the BMSA and discuss its monotonicity properties.

\subsection{$G$-asymmetry for groups generated by Pauli operators}

As a preliminary step, we introduce the $G$-asymmetry, which is a quantity that measures the asymmetry of a state $\rho$ with respect to a given group $G$ \cite{vaccaro2008tradeoff,gour2009measuring}. While originally introduced in the quantum-information literature, the $G$-asymmetry has recently attracted significant attention in the context of quantum many-body physics, especially in connection with the study of symmetry-breaking phenomena~\cite{marvian2014extending,ares2023entanglement,joshi2024observing}. 

Formally, the $G$-asymmetry is a monotone in the resource theory of $G$-frameness~\cite{chitambar2019quantum}. In this framework, the resource is given by symmetric states and free operations are the ones which are symmetric with respect to $G$. In the following, we will not need to delve into the details of the resource theory of $G$-frameness. Therefore, we will simply give the definition of $G$-asymmetry and refer the interested reader to Refs.~\cite{vaccaro2008tradeoff,gour2009measuring, chitambar2019quantum} for a thorough introduction. 

Given a unitary representation of the discrete group $G$ on the Hilbert space $\mathcal{H}$, we denote by $U_g$ the unitary operator associated with $g\in G$. Introducing the  symmetrized state 
\begin{equation}
    \mathcal{G}_G(\rho) = \frac{1}{|G|} \sum_{g \in G} U_g \rho U_g^{\dagger},
\end{equation}
the $G$-asymmetry is defined as \cite{vaccaro2008tradeoff}
\begin{equation} \label{eq:ent_asymmetry}
\begin{split}
    A_G (\rho) &= S(\rho \Vert \mathcal{G}_G(\rho))  \\
    &= S(\mathcal{G}_G(\rho)) - S(\rho),
\end{split}
\end{equation}
where the relative entropy $S(\rho \Vert \sigma)$ is given in Eq.~\eqref{eq:relative_entropy}. 
It is immediate to see that $\mathcal{G}_G(\rho)$  is $G$-invariant. In fact, it is known that $\mathcal{G}_G(\rho)$ is the $G$-invariant state which is closest to $\rho$ \cite{gour2009measuring}. In other words, the $G$-asymmetry is the relative entropy between a state and the nearest $G$-invariant state. In the following, we will also consider the R\'enyi generalizations defined as
\begin{equation} \label{eq:renyi_asymmetry}
    A_{G,\alpha}(\rho) = S_\alpha(\mathcal{G}_G(\rho)) - S_\alpha(\rho),
\end{equation}
where $S_\alpha$ is the R\'enyi entropy.

In this work, we will take $G$ to be a group generated by mutually commuting Pauli operators $P_1, P_2, ..., P_k \in \mathcal{P}_N$. In this case, the symmetrized state $\mathcal{G}_G(\rho)$ takes a simple form: denoting by $G^{\perp}$ the subset of $\mathcal{P}_N$ which commute with all $g\in G$, we obtain
\begin{equation} \label{eq:rho_sym}
    \mathcal{G}_G(\rho) = \frac{1}{2^N}  \sum_{P\in G^{\perp}} \Tr[\rho P] P ,
\end{equation}
see Appendix~\ref{sec:properties_BMSA} for a proof. Namely, the Pauli representation of $\mathcal{G}_G(\rho)$ is obtained from that of $\rho$ by keeping only the Pauli strings in $G^{\perp}$ and discarding the other terms. Exploiting this result, we see that the $G$-asymmetry for $\alpha=2$ takes the nice form
\begin{equation} \label{eq:s2_pauli}
    A_{G,2}(\rho) = -\operatorname{log}_2 \frac{\sum_{P\in G^{\perp}} |\Tr[\rho P]|^2 }{ \sum_{P\in \mathcal{P}_N} |\Tr[\rho P]|^2}.
\end{equation}
The ratio in the right-hand side can be interpreted as the probability that a Pauli string sampled over the Pauli set $\mathcal{P}_{N}$ with probability $\propto |\Tr[\rho P]|^2/{2^N}$ is an element of $G^{\perp}$, \emph{i.e.} it commutes with all Pauli strings in $G$.

\subsection{The basis-minimised stabilizerness asymmetry}

We are finally in a position to introduce the BMSA. Let $| \psi \rangle$ be a pure state, with density matrix $\rho=| \psi \rangle  \langle \psi |$ and take $\ket{s}$ to be a pure stabilizer state. Recalling that $G(\ket{s})$ denotes the stabilizer group of $\ket{s}$, we consider the symmetrized state with respect to $G(|s \rangle)$, namely
\begin{align}\label{eq:intermediate_1}
    \mathcal{G}_{G(| s \rangle)}(\rho) &= \frac{1}{2^N} \sum_{P \in G(| s \rangle)} P \rho P^{-1}\nonumber\\
    &= \frac{1}{2^N} \sum_{P \in G(| s \rangle)} \Tr[\rho P] P,
\end{align}
where the second equality follows from~\eqref{eq:rho_sym}. From Eq.~\eqref{eq:intermediate_1}, it is immediate to see that the $G$-asymmetry~\eqref{eq:ent_asymmetry} is vanishing if and only if $\rho$ is stabilized by $G(\ket{s})$, \emph{i.e.} $P\rho P=\rho$ for all $P\in G(\ket{s})$. In turn, this implies that $\ket{\psi}$ is also a stabilizer state, whose stabilizer group generators coincide with those of $\ket{s}$, up to a phase. 

This observation leads to a natural candidate for a measure of nonstabilizerness, by minimizing the $G$-asymmetry over the set of all possible pure stabilizer states $\text{PSTAB}_N$. Formally, we define the BMSA as 
\begin{equation} \label{eq:smin}
\begin{split}
    \mathcal{A}_1(\ket{\psi}) &= \min_{|s \rangle \in \text{PSTAB}_N} A_{G(\ket{s})}(\rho) \\ 
    &= \min_{|s \rangle \in \text{PSTAB}_N}  S(\mathcal{G}_{G(\ket{s})}(\rho)),
\end{split}
\end{equation}
together with its R\'enyi version
\begin{equation} \label{eq:smin_renyi}
\begin{split}
    \mathcal{A}_{\alpha}(\ket{\psi}) &= \min_{|s \rangle \in \text{PSTAB}_N}  S_\alpha(\mathcal{G}_{G(\ket{s})}(\rho)).
\end{split}
\end{equation}
In the following, we will also write $   \mathcal{A}$ instead of $\mathcal{A}_1$ when this does not generate confusion. The rest of this section is devoted to discussing the properties of the BMSA.

First, it follows from elementary considerations that $\mathcal{A}_{\alpha}(\ket{\psi})$ has the following properties:
\begin{enumerate}
    \item Faithfulness: $\mathcal{A}_{\alpha}(\ket{\psi}) =0$ iff $\ket{\psi} \in \text{PSTAB}_N$
    \item Invariance under Clifford unitaries: $\mathcal{A}_{\alpha}(C\ket{\psi})=\mathcal{A}_{\alpha}(\ket{\psi})$ for all Clifford gates $C \in \mathcal{C}_N$
    \item Subadditivity: 
   \begin{equation}
       \mathcal{A}_{\alpha}(\ket{\psi}_{A} \otimes\ket{\phi}_{B}) \leq \mathcal{A}_{\alpha}(\ket{\psi}_{A})+\mathcal{A}_{\alpha}(\ket{\phi}_{B} )
   \end{equation} 
    \item The BMSA is lower and upper bounded by
    \begin{equation}
        0\leq \mathcal{A}_\alpha(\ket{\psi}) \leq N \log 2\,.
    \end{equation}
\end{enumerate}
As we will discuss later, the sub-additivity inequality is generally strict. However, we will provide numerical evidence that $\mathcal{A}_{\alpha}$ increases approximately linearly in $N$ for tensor products $\ket{\phi}^{\otimes N}$, for any $\ket{\phi}$.

Next, when $\alpha=1$, we can prove the anticipated stronger result: $\mathcal{A}_1$ is a strong monotone (and thus also a monotone) for the resource theory of nonstabilizerness. The proof of this statement is technical and relies on known properties of the $G$-asymmetry. We report it in Appendix~\ref{sec:properties_BMSA}. 
We now proceed to discuss whether R\'enyi BMSAs are strong monotones for $\alpha\neq1$. 
By bounding with $D_\mathrm{min}$ (which we will prove below) and the fact that $D_\mathrm{min}$ is not a strong monotone \cite{haug2023stabilizer}, we can show that the R\'enyi BMSA are not strong monotones for $\alpha \geq 2$. It remains an open question whether the R\'enyi BMSA for $1 < \alpha < 2$ satisfy the strong monotonicity property. For $\alpha<1$, we can show that $\alpha=1/2$ is also a strong monotone, by connecting to known properties of coherence measures (see Sec. \ref{sec:coherence}). We leave the question of strong monotonicity for $1<\alpha<2$, and $\alpha<1, \alpha \neq 1/2$ as an open problem. In Appendix~\ref{sec:renyi_BMSA} we discuss the properties of the R\'enyi-$2$ BMSA in more details.

Importantly, computing the BMSA involves a minimization procedure, so it is harder than the stabilizer R\'enyi entropies. Still, the minimization in  Eq. \eqref{eq:smin} is only over pure stabilizer states, making it significantly more efficient to be computed with respect to previously known strong monotones such as the relative entropy of magic~\eqref{eq:relative_entropy_of_magic} or the RoM~\eqref{eq:robustness}. Note in particular that due to the freedom in the signs of the stabilizer generators, the minimization in Eq. \eqref{eq:smin} is only over $|\text{PSTAB}_N|/2^N$ stabilizer basis. In fact, in Sec.~\ref{sec:numerical_results} we will discuss methods to evaluate it in the many-body setting, providing explicit numerical data in several cases up to $N\sim 10$. 

Before leaving this section, it is important to mention that the BMSA has been so far defined only for pure states. However, following a standard \emph{convex-roof} construction, we can extend the BMSA to mixed states as
\begin{equation} \label{eq:smin_mixed}
    \mathcal{A}_1(\rho) = \min_{\{p_i, \rho_i \}}  \sum_i p_i  \mathcal{A}_1(\rho_i) ,
\end{equation}
where the minimum is taken over all possible convex decompositions of $\rho$: $\rho = \sum_i p_i \rho_i$, with $\rho_i$ pure. It can be shown that this extension has the same properties as the pure-state version. However, while providing a valid monotone, Eq.~\eqref{eq:smin_mixed} remains hard to be computed in practice.

\section{Relation with other monotones}
\label{sec:relations}

\subsection{Relation with magic monotones}

In this section we discuss the relation between the BMSA and other monotones known in the literature.

\subsubsection{Relation to basis-minimized measurement entropy} \label{sec:bmme}

As anticipated, we first show that the BSMA coincides with the basis-minimized measurement entropy. This quantity has been recently introduced in Ref.~\cite{niroula2023phase}, which also reported its experimental detection. It is defined. by
\begin{equation} \label{eq:bmme}
    z^*(| \psi \rangle) = \min_{C \in \mathcal{C}_N} S^{\mathrm{part}}_1(C^\dagger | \psi \rangle),
\end{equation}
where $S^{\mathrm{part}}_q (| \psi \rangle)$ are the participation entropies (PEs)
\begin{equation}
    S^{\mathrm{part}}_q (| \psi \rangle) = \frac{1}{1-q} \log \sum_\sigma \vert \langle \sigma | \psi \rangle\vert ^{2q},
\end{equation}
and $\mathcal{C}_N$ is the Clifford group for $N$ qubits. 
We can also consider the R\'enyi generalization of Eq. \eqref{eq:bmme}
\begin{equation} 
    z_q^*(| \psi \rangle) = \min_{C \in \mathcal{C}_N} S^{\mathrm{part}}_q(C^\dagger | \psi \rangle),
\end{equation}
which are also measures of magic. 

We now show that for any $\alpha \geq 0$, we have
\begin{equation} \label{eq:equal_bmme}
    \mathcal{A}_\alpha(\ket{\psi}) = z^*_\alpha(| \psi \rangle).
\end{equation}
To see this, let $|s \rangle=|0 \rangle^{\otimes N}$. We can write
\begin{equation}
    \mathcal{G}_{G(|0 \rangle^{\otimes N})}(\rho) = \frac{1}{2^N} \sum_{P\in \mathcal{P}_z} \Tr[\rho P] P,
\end{equation}
where $\mathcal{P}_z$ is the group of Pauli operators which only contain $I$ and $Z$ operators. One can see that $\mathcal{G}_{G(|0 \rangle^{\otimes N})}(\rho)$ is exactly the diagonal part of $\rho$. Therefore, the R\'enyi entropy $S_\alpha(\mathcal{G}_{G(|0 \rangle^{\otimes N})}(\rho))$ is simply the Shannon entropy of the diagonal elements of $\rho$, which is exactly the participation entropies $S^{\mathrm{part}}_\alpha(\ket{\psi})$. Moreover, for a generic stabilizer state $\ket{s} \neq |0 \rangle^{\otimes N}$, we note that $\mathcal{G}_{G(\ket{S})}(\rho)$ can be brought to a diagonal form by a Clifford unitary (see also Sec. \ref{sec:coherence}). This immediately implies Eq. \eqref{eq:equal_bmme}. 

Importantly, Ref.~\cite{niroula2023phase} only established elementary properties of the basis minimised measurement entropy. Now, the identification in Eq.~\eqref{eq:equal_bmme} allows us to establish that it is a strong monotone of nonstabilizerness. This fact is particularly interesting, given the experimental relevance of the basis minimised measurement entropy~\cite{niroula2023phase}.

\subsubsection{Relation to relative entropy of magic}
Next, we show that the BMSA is related to the relative entropy of magic~\eqref{eq:relative_entropy} by the inequality
\begin{equation} \label{eq:relation_rem}
    r_{\mathcal{M}}(\rho) \leq \mathcal{A}_1(\rho) .
\end{equation}
To see this, we will first show that $\mathcal{G}_{G(\ket{s})}(\rho) \in {\rm STAB}$. We recall that $\mathcal{G}_{G(\ket{s})}(\rho)$ can be brought to a diagonal form by a Clifford unitary. Moreover, its diagonal elements take values between 0 and 1. This implies that the diagonal matrix can be written as a convex combination of the $2^N$ stabilizer states with the same unsigned stabilizer group $G(\ket{s})$. 

Now, let 
\begin{equation}
\mathcal{A}_1(\rho) = \sum_i p_i S(\rho_i \Vert \mathcal{G}_{G(\ket{s_i})}(\rho))\,,   
\end{equation}
 where $\rho = \sum_i p_i \rho_i$. Here, $\rho_i$ is pure and achievies the minimum in~\eqref{eq:smin_mixed}, while $\ket{s_i}$ are the stabilizer states achieving the minimum in the definition of the BMSA for pure states. We have
\begin{align}
    \mathcal{A}(\rho) &= \sum_i p_i S(\rho_i \Vert \mathcal{G}_{G(\ket{s_i})}(\rho)) \nonumber\\
    &\geq S\left(\sum_i p_i \rho_i \Vert \sum_i p_i \mathcal{G}_{G(\ket{s_i})} (\rho)\right) \geq r_{\mathcal{M}}(\rho).
\end{align}
The second inequality is due to the joint convexity of relative entropy.

\subsubsection{Relation to SRE}
It can be shown that $\mathcal{A}_{\alpha}(\ket{\psi})$ is related to the SRE by the following inequality: 
\begin{equation} \label{eq:relation_sre}
    M_{n}(|\psi \rangle) \leq \frac{n}{n-1} \mathcal{A}_{\alpha}(\ket{\psi}) \quad (n >1, \alpha \leq 2).
\end{equation}
The proof is similar to the proof for a similar inequality between the SRE and min-relative entropy of magic \cite{haug2023stabilizer}.

We further find the following bound between SRE and BMSA
\begin{equation}\label{eq:second_inequality}
    2\mathcal{A}_\alpha(\ket{\psi})\geq M_n(\ket{\psi}) \quad (n\geq1/2, \alpha\leq1/2)
\end{equation}
for any $\alpha\leq1/2$ and $n\geq1/2$. Note however that the bound is weaker than Eq.~\eqref{eq:relation_sre} for $n>2$. Eq.~\eqref{eq:second_inequality} is proved in Appendix~\ref{sec:properties_BMSA}.

Finally, using the recent result of Ref.~\cite{arunachalam2024notepolynomialtimetoleranttesting}, one can show there exist a constant $C>1$ (i.e. independent of $N$) such that
\begin{equation}
    \mathcal{A}_\alpha(\ket{\psi})\leq 2 C M_n(\ket{\psi})  \quad (n\leq 2, \alpha\geq 2 )\,.
\end{equation}
This allows us to give an explicit upper and lower bound on $\mathcal{A}_2$ via the SRE
\begin{equation}\label{eq:to_prove_1}
\frac{1}{2}M_2(\ket{\psi}) \leq \mathcal{A}_2(\ket{\psi})\leq 2 C M_2(\ket{\psi})\,.
\end{equation}
Note that this also solves the open problem of Ref.~\cite{haug2023stabilizer,haug2023efficient}, showing that the SRE is indeed a qubit-number independent upper bound for the min-relative entropy of magic $D_\text{min}$ via
\begin{equation}\label{eq:to_prove_2}
    \frac{m-1}{2m} M_m(\ket{\psi}) \leq D_\text{min}(\ket{\psi})\leq CM_n(\ket{\psi})
\end{equation}
for $(n\leq2, m\geq1 )$, where for completeness we also added the previously proven lower bound~\cite{haug2023stabilizer}. We prove Eq.~\eqref{eq:to_prove_1} in Appendix~\ref{sec:proof_inequality}.

\subsubsection{Relation to stabilizer nullity}
Next, it can be shown that $\mathcal{A}_{\alpha}(\ket{\psi})$ is related to the stabilizer nullity $\nu(|\psi \rangle)$ by the following inequality:
\begin{equation} \label{eq:relation_nullity}
    \mathcal{A}_{\alpha}(\ket{\psi})/\ln 2 \leq \nu(|\psi \rangle) .
\end{equation}
To see this, we use the fact that, for a state with stabilizer nullity $\nu$, there exists a Clifford unitary $C$ such that $C\ket{\psi} = \ket{0}^{N-\nu} \ket{\phi}$, where $\ket{\phi}$ is a pure state of $\nu$ qubits~\cite{beverland2020lower}. 
Therefore,
\begin{equation}
    \mathcal{A}_\alpha(\ket{\psi}) = \mathcal{A}_\alpha(\ket{0}^{N-\nu} \ket{\phi}) \leq  S^{\mathrm{part}}_\alpha(\ket{\phi}) \leq  \nu \ln 2.
\end{equation}

\subsubsection{Relation to min-relative entropy of magic}
Furthermore, it can be shown that the min-relative entropy of magic $D_{\mathrm{min}}$ provides both lower and upper bound to $\mathcal{A}_{\alpha}(\ket{\psi})$:
\begin{equation} \label{eq:upper_dmin}
     \mathcal{A}_{\alpha}(\ket{\psi}) \leq 2 D_{\text{min}}(|\psi \rangle) \quad (\alpha \geq 2),
\end{equation}
and
\begin{equation} \label{eq:lower_dmin}
    D_{\text{min}}(|\psi \rangle)  \leq \mathcal{A}_{\alpha}(\ket{\psi})  \quad (\alpha \geq 0).
\end{equation}
For a pure state $| \psi \rangle$, let $| s \rangle$ be the stabilizer state with the largest overlap with $| \psi \rangle$, i.e., $D_{\text{min}}(| \psi \rangle) = -\log |\langle s | \psi \rangle|^2$. We have
\begin{equation}
    \begin{split}
        \frac{1}{2^N} \sum_{P\in G(| s \rangle)} |\Tr[\rho P]|^2  &\geq \frac{1}{2^{2N}} \left( \sum_{P\in G(| s \rangle)} |\Tr[\rho P]| \right)^2\\
        & \geq \frac{1}{2^{2N}} \left| \sum_{P\in G(| s \rangle)} \Tr[\rho P] \right|^2\\
        & =  \left| \Tr[\rho |s \rangle \langle s |] \right|^2\\
        & =  \left|  \langle s | \psi \rangle \right|^4.\\
    \end{split}
\end{equation}
In the first line we used Jensen's inequality, while in the second line we used the triangle inequality. Taking the logarithm on both sides, we obtain
\begin{equation} \label{eq:a}
    S_2(\mathcal{G}_{G(\ket{s})}(\rho)) \leq 2 D_{\text{min}}(| \psi \rangle).
\end{equation}
Finally, using the hierarchy of R\'enyi entropies $S_a \leq S_b$ for $a \geq b$, and combining with Eq. \eqref{eq:a}, we obtain the upper bound in Eq. \eqref{eq:upper_dmin}. 

For the lower bound in Eq. \eqref{eq:lower_dmin}, we first show that, in fact, the min-relative entropy of magic is related in a more direct way to the BMSA via the limit
\begin{equation} \label{eq:dmin_a}
    \lim_{\alpha \to \infty} \mathcal{A}_\alpha(\ket{\psi}) =  D_{\mathrm{min}} (\ket{\psi})
\end{equation}
where $D_{\mathrm{min}}$ is the min-relative entropy of magic. Indeed, by Eq. \eqref{eq:equal_bmme}, we have
\begin{align}
    \lim_{\alpha \to \infty}  \mathcal{A}_\alpha(\ket{\psi}) &= \lim_{\alpha \to \infty}  z^*_\alpha(| \psi \rangle)\nonumber\\
    =& \min_{C \in \mathcal{C}_N} -\log \max_{\sigma} \vert \bra{\sigma} C^\dagger | \psi \rangle \vert^2, 
\end{align}
where $\sigma$ denotes the computational basis states. Since $C \ket{\sigma}$ encompasses the set of pure stabilizer states, it follows that
\begin{align}
    \lim_{\alpha \to \infty}  \mathcal{A}_\alpha(\ket{\psi}) &=  \min_{|S \rangle \in \text{PSTAB}_N} -\log  \vert \bra{s}  \psi \rangle \vert^2 \nonumber\\
    &=  D_{\mathrm{min}} (\ket{\psi}).
\end{align}
Then, using Eq. \eqref{eq:dmin_a} and the hierarchy of R\'enyi entropies, we obtain the lower bound in Eq. \eqref{eq:lower_dmin}.

\subsubsection{Relation to max-relative entropy of magic}
Another important magic monotone is the max-relative entropy of magic, which for pure states is equivalent to the logarithm of the stabilizer extent~\cite{bravyi2019simulation,liu2022many}. It can be written as a linear program 
\begin{equation}
    D_\text{max}(\ket{\psi})=2\ln[\min_{x} (\sum_{i}\vert x_i\vert) \text{ s.t. } \ket{\psi}=\sum_i x_i\ket{s_i}]
\end{equation}
with $x\in \mathbb{C}^{\vert \text{STAB}_N\vert}$ and $\ket{s_i}\in \text{STAB}_N$.

Similarly, the basis-minimized nonstabilizerness  asymmetry for $\alpha=1/2$ can be written in a similar form as
\begin{align}\label{eq:asymmetryProgram}
    \mathcal{A}_{1/2}(\ket{\psi})&=2\ln[\min_{x,\ket{s}\in\text{STAB}_N} \sum_i\vert x_i\vert \text{ s.t. } \ket{\psi}\nonumber\\
    &=\sum_i x_i d_i^{s}\ket{s}]
\end{align}
where $d_i^{s}$ is one of the $2^N$ destabilizers of $\ket{s}$ where $\{d_i^{s}\ket{s}\}_i$ is a basis. As all $d_i^{s}\ket{s}$ are stabilizer states and thus a subset of $\text{STAB}_N$~\cite{yoder2012generalization}, this implies $\mathcal{A}_{1/2}(\ket{\psi})\geq D_\text{max}(\ket{\psi})$ and together with the hierarchy of R\'enyi entropies
\begin{equation}
    \mathcal{A}_{\alpha}(\ket{\psi})\geq D_\text{max}(\ket{\psi})\,, \qquad \left(\alpha\leq \frac{1}{2}\right)\,.
\end{equation}
Note that the log-free robustness of magic $\text{LR}$ is also is an upper bound to $D_\text{max}$ via $\text{LR}\geq D_\text{max}$~\cite{liu2022many} However, counter-examples show that such inequality with prefactor $1$ cannot exist between $\text{LR}$ and $\mathcal{A}_{1/2}$, i.e. in general $\text{LR}\rlap{\kern.45em$|$}\geq \mathcal{A}_{1/2}$ and $\text{LR}\rlap{\kern.45em$|$}\leq \mathcal{A}_{1/2}$.

\subsubsection{Relation to stabilizer rank}
The stabilizer rank $\chi(\ket{\psi})$ is the size of the smallest set of stabilizer states that can represent a given state $\ket{\psi}$~\cite{bravyi2016trading}
\begin{equation}
    \chi(\ket{\psi})=\text{min}_{x} k \text{ s.t. } \ket{\psi}=\sum_{i=1}^k x_i \ket{s_i}
\end{equation}
with stabilizer states $\ket{s_i}$ and coefficients $x_i\in\mathbb{C}$. 
The BMSA for the case $\alpha=0$ is written as
\begin{equation}
\mathcal{A}_0(\ket{\psi})=\ln(\min_{x,\ket{s}} k \text{ s.t. } \sum_{i=1}^k x_i d_i\ket{s})\,,
\end{equation}
which corresponds to finding the stabilizer basis with minimal support that can represent $\ket{\psi}$.
We can immediately see that $\mathcal{A}_0(\ket{\psi})$ is similar to the problem of finding stabilizer rank, but over a more restricted set of stabilizers (i.e. only over stabilizer basis).
Thus, we have
\begin{equation}
    \mathcal{A}_0(\ket{\psi})\geq \ln(\chi(\ket{\psi}))\,.
\end{equation}

\subsection{Connection with coherence measures} \label{sec:coherence}
The BSMA can also be understood in terms of coherence measures. Indeed, $\mathcal{A}(\ket{\psi})$ is equivalent to minimizing the relative entropy of coherence $C_r$~\cite{baumgratz2014quantifying}, which is a strong coherence monotone, over all possible stabilizer basis. Similarly, for pure states, $\mathcal{A}_{1/2}(\ket{\psi})$ is equivalent to minimizing the $l_1$ norm coherence $C_{l_1}$~\cite{baumgratz2014quantifying}, another strong coherence monotone, over all possible stabilizer basis. It can thus be shown that $\mathcal{A}_{1/2}(\ket{\psi})$ is also a strong magic monotone.

This connection to coherence can also be seen in the following way. Usually, the set of stabilizer states is seen as the convex combinations of pure stabilizer states. One can see that this set is equivalent to the convex combination of incoherent states in stabilizer basis. Specifically, we can write the set of stabilizer states as
\begin{equation}
    \mathrm{STAB} = \left \lbrace \rho : \rho = \sum_j p_j  \tau_j  , \forall j \,p_j \geq 0, \sum_j p_j = 1  \right \rbrace,
\end{equation}
where $\tau_j$ are incoherent states in the stabilizer basis.

It is worth noting that similar construction in the case of entanglement leads to R\'enyi entanglement entropy, which are known to be strong measures of entanglement for $0 \leq \alpha \leq 1$ \cite{vidal2000entanglement}, but not for $\alpha>1$ \cite{horodecki2009quantum}.

\section{Numerical approaches}
\label{sec:numerical_results}

In this section, we present our numerical methods to compute the BMSA. Before that, we discuss a natural way in which we can interpret the BMSA in terms of the computational cost of classically simulating  quantum circuits.

\subsection{BMSA and the computational cost of classical simulations}

As discussed above, given any stabilizer state $\ket{s}$, one can define a stabilizer basis via its destabilizers $d_i$~\cite{yoder2012generalization} to represent any state via 
$\ket{\psi}=\sum_{i=1}^{2^n} c_i d_i\ket{s}$ with coefficients $c_i$. The stabilizerness asymmetry corresponds to the entropy of the squared coefficients $\{\abs{c_i}^2\}_i$, i.e. $A_{G(\ket{s})}(\ket{\psi})=S(\{\abs{c_i}^2\}_i)$, and the BMSA corresponds to the stabilizer basis with minimal entropy, which represents the favorable choice for simulation.

Concretely, let us consider the simulation of circuits composed of Clifford gates and non-Clifford Pauli rotations. We would like to compute the expectation value
\begin{equation}
    \langle O \rangle = \langle 0^{\otimes N} | C_1^\dagger U_1^\dagger \cdots C_D^\dagger U_D^\dagger O C_D U_D \cdots C_1 U_1    | 0^{\otimes N} \rangle 
\end{equation}
of a Pauli operator $O$. Here $C_j$ are Clifford gates and $U_j$ are non-Clifford Pauli rotation gates $U_P(\theta)= \exp(i \theta P)$. Note that, without loss of generality, the circuits above can be reduced to circuits containing only Pauli rotations $U_j$ through circuit compilation. Here, we consider the simulation of the circuits by evolving the state through the gates (Schr\"odinger picture).

The action of the non-Clifford Pauli rotation operator $U_P(\theta)= \exp(i \theta P) = \cos(\theta) + i \sin(\theta) P $ on a stabilizer state is given by
\begin{equation}
    U_P(\theta) \ket{s} = 
    \begin{cases}
    \exp(i \theta \lambda_s) \ket{s} , & P \in G(\ket{s}) \\
    \cos(\theta)\ket{s} +  i \sin(\theta) \ket{s'}, & \text{otherwise }
    \end{cases}
\end{equation}
where $\ket{s'} = P \ket{s}$ is a stabilizer state in the same stabilizer basis as $\ket{s}$ and $\lambda_s = \pm 1$. Note that $\ket{s'}$ can be obtained efficiently through the standard stabilizer simulation technique.

Therefore, if we apply a sequence of Pauli rotation gates to the state $\ket{0^{\otimes N}}$, each gate would increase the number of terms in the stabilizer basis at most by a factor of 2. We can thus perform a classical simulation of such circuit by simply keeping track of the coefficients in a given stabilizer basis (and possibly minimized over). The simulation can be done approximately by truncating all coefficients below a given error threshold $|c_i|<\epsilon$ to control the computational cost. The BMSA is related to the accuracy of such truncation: if $\mathcal{A}_1(\ket{\psi})$ scales as $O(N)$, then for a given small $\epsilon$ the required number of coefficients that needs to be retained scales as $O(\exp(N))$. Following an argument in Ref.~\cite{schuch2008entropy}, this can be seen using Fannes’ inequality~\cite{audenaert2007sharp} and the fact that $||\mathcal{G}_G(\ket{\psi})-\mathcal{G}_G(\ket{\phi})||_1\leq ||\ket{\psi}-\ket{\phi}||_1$, which follows from the monotonicity of the trace distance with respect to quantum channels~\cite{nielsen2011quantum}. This simulation technique is very similar, and in fact can be seen as the Schr\"odinger version of the sparse Pauli dynamics method \cite{begui2024fast,begui2024realtime}. A related technique involves approximating the state in a stabilizer basis as matrix product states~\cite{masot2024stabilizer}.

\subsection{Brute-force computation from Pauli vector}

For a small number of qubits, the BMSA can be computed directly by minimizing over the full stabilizer group, as we explain below. The method is adapted from the method in \cite{hamaguchi2024handbook}.  

Consider a single qubit Pauli operator $P_{\alpha,\alpha^{\prime}}=i^{\alpha \alpha^{\prime}} X^{\alpha} Z^{ \alpha^\prime}$. For $N$ qubits, the Pauli strings can then be written as
\begin{equation}
P_{\bold{\alpha},\bold{\alpha^{\prime}}}=P_{\alpha_1,\alpha_1^{\prime}} P_{\alpha_2,\alpha_2^{\prime}} \cdots  P_{\alpha_N,\alpha_N^{\prime}}
\end{equation}
For a quantum state $\rho$, we denote the Pauli vector by $b_{\boldsymbol{\alpha},\boldsymbol{\alpha'}} = \Tr [\rho P_{\boldsymbol{\alpha},\boldsymbol{\alpha'}}]$. Let $G$ be a stabilizer group with generators $g_1, \cdots,g_N$, whereby $g_i = P_{\boldsymbol{\alpha}_i,\boldsymbol{\alpha'}_i}$.  We construct the vector $c_\mathbf{u} = \Tr [\rho g_1^{u_1} \cdots g_N^{u_N} ] = b_{\mathbf{v}, \mathbf{v'}}$, where $\mathbf{v} = \oplus_i u_i\boldsymbol{\alpha}_i$ and $\mathbf{v'} = \oplus_i u_i\boldsymbol{\alpha'}_i$. This operation implements the twirling operation in Eq. \eqref{eq:intermediate_1}, such that $c_\mathbf{u}$ stores the Pauli vector of $\mathcal{G}_{G(\ket{s})(\rho)}$. Then, we compute $\mathbf{d} =  H^{\otimes N} \mathbf{c} / 2^N$, where $H=[1, 1;1 ,-1]$ is the unnormalized Hadamard matrix. The vector $\mathbf{d}$ stores the coefficients in the stabilizer basis associated to $G(\ket{s})$. Finally, we compute the Shannon entropies of the elements of $\bold{d}$. 

The computation allows us to compute $A_{G\ket{s}(\rho)}$ for a given stabilizer state $\ket{s}$. The numerical computation of $\mathcal{A}(\rho)$ is completed by minimizing over all possible stabilizer states $\ket{s}$, which can be done for small system sizes. In practice, we are able to follow this approach up to $N=5$.

\subsection{Exact computation via branch and bound} \label{sec:branch_and_bound}

Going beyond the brute force method described previously, we explain a more sophisticated approach to compute the BMSA exactly. In particular, we adapt the branch and bound method recently introduced in \cite{hamaguchi2024faster} to compute the BMSA. First of all, we enumerate the stabilizer states as follows. Let $\mathbb{F}_2$ be the finite field with two elements.
  For all $k \in \{1,\dots,N\}$, we define
  \begin{align*}
    \mathcal{Q}_k &= \left\{ Q \bigm| Q \in \mathbb{F}_2^{k \times k} \textrm{ is an upper triangular matrix} \right\},\\
    \mathcal{R}_k &= \left\{ R \bigm| R \in \mathbb{F}_2^{N \times k}, \textrm{$\mathrm{rank}(R)=k$} \right\},\\
    \mathcal{T}_R &= \left\{ \mathrlap{\:t}\phantom{Q} \bigm| \mathrlap{\:t}\phantom{Q} \in   \mathbb{F}_2^N / \Im(R) \right\}. 
  \end{align*}
  We also define the set of states ${\rm PSTAB}_{N,k}$ as
  \begin{equation}\label{eq:stabilizerStateStandardForm}
  \begin{split}
    {\rm PSTAB}_{N,k} =
    \left\{ \frac{1}{2^{k/2}} \sum_{x=0}^{2^k-1}(-1)^{x^\top Q x} i^{c^\top x}\ket{R x+t} \right. 
    \\
    \left. \bigm| Q \in \mathcal{Q}_k, c \in \mathbb{F}_2^k, R \in \mathcal{R}_k, t\in \mathcal{T}_R \vphantom{\frac{1}{2^{k/2}}} \right\},
    \end{split}
  \end{equation}
  and define ${\rm PSTAB}_{N,0} = \left\{ \ket{t} \bigm| t \in \mathbb{F}_2^{N} \right\}$.
  Then, we have $\bigcup_{k=0}^{N} {\rm PSTAB}_{N,k} = {\rm PSTAB}_N$. \cite{VandenNest2010,Struchalin2021,hamaguchi2024faster}

  The number of $Q,c,R,t$ is
  $2^{k(k+1)/2},2^k,$ ${N \brack k}_2$ $,2^{N-k}$, respectively \cite{hamaguchi2024faster}. Here, ${N \brack k}_2$ is a q-binomial coefficient with $q=2$ \cite{Kac2002}. Let $Q=Q_{\mathrm{d}} + Q'$, where $Q_{\mathrm{d}}$ is the diagonal part of $Q$. It can be shown that the $2^N$ stabilizer states characterized by the same $Q',c,R$ form an orthonormal basis. For a quantum state $\ket{\psi}$, we denote the squared overlap with the stabilizer state $\ket{\phi} = \frac{1}{2^{k/2}}\sum_{x=0}^{2^k-1} (-1)^{x^\top Q x} i^{c^\top x} \ket{Rx+t}$ by $p_{Q',c,R}(Q_\mathrm{d},t) = \abs{\braket{\phi}{\psi}}^2$, which forms a probability distribution, i.e., $\sum_{Q_\mathrm{d},t}p_{Q',c,R}(Q_\mathrm{d},t)=1$. 

  The BMSA can be written as 
  \begin{equation}
    \label{eq:bmsa2}\mathcal{A}_\alpha(\ket{\psi}) = \min_k \min_{Q',c,R}  S_\alpha(\{p_{Q',c,R}(Q_\mathrm{d},t) \}_{Q_\mathrm{d},t}) 
  \end{equation}

  Next, we derive a bound for a given $R$. We denote $q_{R}(Q_\mathrm{d},t) = \abs{\braket{\phi'}{\psi'}}^2$ where $\ket{\phi'} = \frac{1}{2^{k/2}}\sum_{x=0}^{2^k-1} (-1)^{ Q_{\mathrm{d}} x}  \ket{Rx+t}$ and we construct a state $\ket{\psi'}$ whose coefficients are given by $\braket{x}{\psi'} = \abs{\braket{x}{\psi}}$. It can be shown that, for fixed $R$ and integer $\alpha>1$,
  \begin{equation} \label{eq:bound_r}
      \min_{Q',c}  S_\alpha(\{p_{Q',c,R}(Q_\mathrm{d},t) \}_{Q_\mathrm{d},t}) \geq S_\alpha(\{q_{R}(Q_\mathrm{d},t) \}_{Q_\mathrm{d},t}) 
  \end{equation}
  The derivation can be found in Appendix~\ref{sec:deriv_bound}. This bound can be utilized to implement the branch and bound method, as discussed in detail in Ref.~\cite{hamaguchi2024faster}. The key idea is that, since we are solving for a minimum, any branch whose lower bound is higher than the current best solution cannot contain the global minimum. This allows us to skip this branch and significantly reduce the search space compared to the brute-force approach. Additionally,  the computation can be further simplified by exploiting the special structures of the wavefunction. We discuss this in Appendix~\ref{sec:real_states}, where we consider both the case of wavefunctions with real amplitudes and positive real amplitudes.
  
We employed this method to examine the additivity properties of the BSMA. Note that, while we are unable to prove Eq. \eqref{eq:bound_r} for the case $\alpha=1$, we conjecture that the lower bound remains valid, as we have not found any counter-examples. In the following, we applied the branch and bound method while assuming that the inequality holds for $\alpha=1$.  Specifically, we considered product states of the form $|\Psi(\theta)\rangle = (|0\rangle + e^{i\theta}|1\rangle)^{\otimes N}$, with two distinct choices of $\theta$: $\theta_1 = \frac{\pi}{4}$, corresponding to the $T$-state, and $\theta_2 = \frac{\pi}{8}$. Figure~\ref{fig:product_states} displays the computational results as a function of system size $N$. Our findings indicate that BMSA scales linearly with $N$ in both cases. For the case of $\theta_2$ we were able to access larger system sizes compared to $\theta_1$. This is attributed to the lower BMSA values for $\theta_2$, which enable the algorithm to more efficiently prune the search space and find the optimal solution. Additionally, in the asymptotic limit of large $N$, we demonstrate linearity in $N$ using the argument detailed in Section~\ref{sec:bounds}.

\begin{figure}
    \centering 
    \includegraphics[width=1.\linewidth]{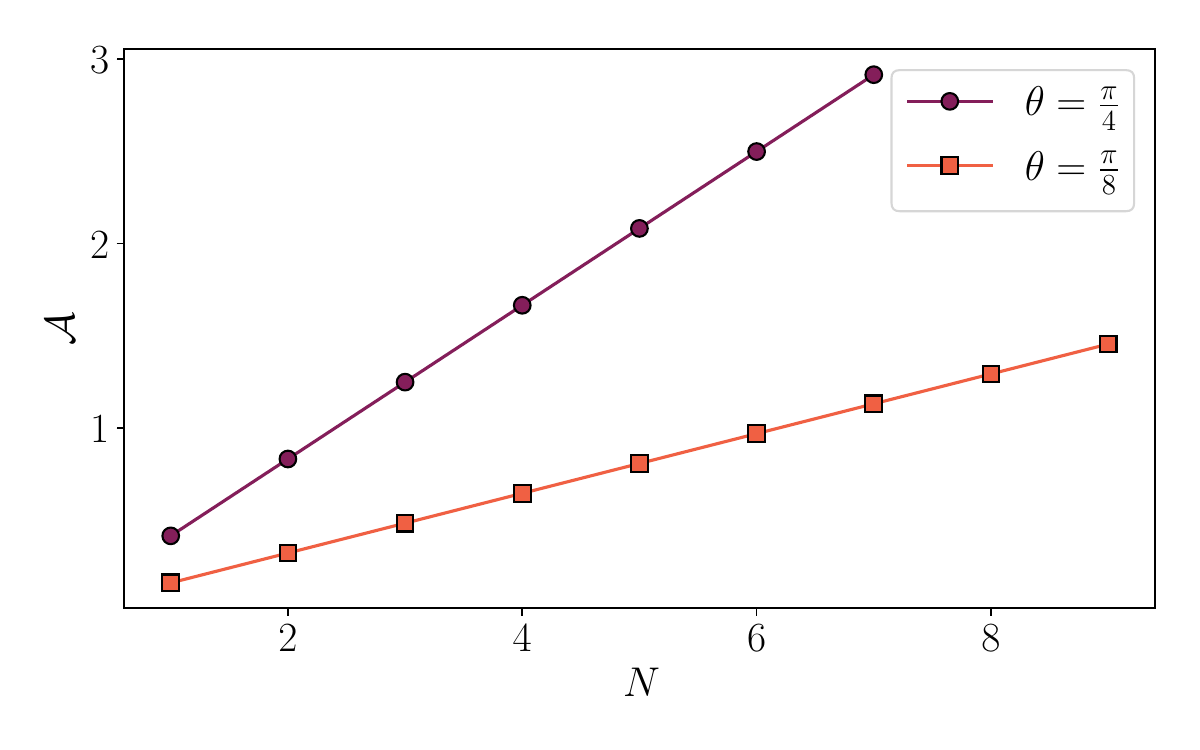}
    \caption{Scaling of BMSA with system size $N$ for two different product states, both of the form $|\Psi(\theta)\rangle = (|0\rangle + e^{i\theta}|1\rangle)^{\otimes N}$, with $\theta_1 = \frac{\pi}{4}$, $\theta_2 = \frac{\pi}{8}$. }
    \label{fig:product_states}
\end{figure}

\subsection{Estimation via minimization of participation entropy}\label{sec:min_part_entropy}
The exact computation of magic asymmetry has a cost that is super-exponential with the system size. To reduce this cost, here we adapt the algorithm proposed in \cite{qian2024augmenting}. Using the relation described in Eq.(\ref{eq:equal_bmme}), we can find the asymmetry efficiently by minimizing the participation entropy. Given a state, we can apply a Clifford circuit to the wave function $|\psi\rangle$ to reduce the participation entropy. By acting the Clifford circuits, we obtain the wave function
\begin{equation}
    |\tilde{\psi} \rangle=U_C |\psi\rangle
\end{equation}
where $U_C$ denotes the Clifford circuits. By carefully selecting appropriate Clifford transformations, the resulting wave function can exhibit reduced participation entropy compared to the original wave function. The minimization routine constructs the optimal Clifford operator iteratively by sweeping over the chain and minimize the entanglement across neighboring sites of the chain using two-qubit Clifford unitaries. For each nearest-neighbor pair, the participation entropy is computed after applying each two-qubit Clifford unitary $C \in \mathcal{C}_2$. Let $\mathcal{C}_z$ be the group of Clifford unitaries that map Pauli strings in $\mathcal{P}_z$ to itself. Given that Clifford unitaries in $\mathcal{C}_z$ do not alter the participation entropy, only $|\mathcal{C}_2|/|\mathcal{C}_z|=15$ gates need to be considered at each step where $|\mathcal{C}_2|=11520$ and $|\mathcal{C}_z|=768$. The selection of a specific $C$ is performed using a heat bath method, where we introduce a fictitious temperature $T$, and the probabilities for each $C_i \in \mathcal{C}_2/\mathcal{C}_z$ are computed as
\begin{equation}
    p_{C_i} = \mathrm{e}^{-(S^{\mathrm{part}}_\alpha (C_i C_\mathrm{ref}| \psi \rangle) - S^{\mathrm{part}}_\alpha (C_\mathrm{ref} |\psi \rangle))/T}
\end{equation}
up to normalization. Here $C_\mathrm{ref}$ represents the Clifford gate that minimizes the participation entropy at the considered step. If the gate $C_i$ is selected, $C_\mathrm{ref}$ is updated as $C_\mathrm{ref} \to C_i C_\mathrm{ref}$. The temperature $T$ is decreased at each sweep, ranging on a logarithmic grid from $10^{-1}$ to $10^{-4}$. This algorithm resembles the simulated annealing \cite{kirkpatrick1983optimization}, a method often used to search for global minimum in optimization problems.

\begin{figure}
    \centering
    \includegraphics[width=1\linewidth]{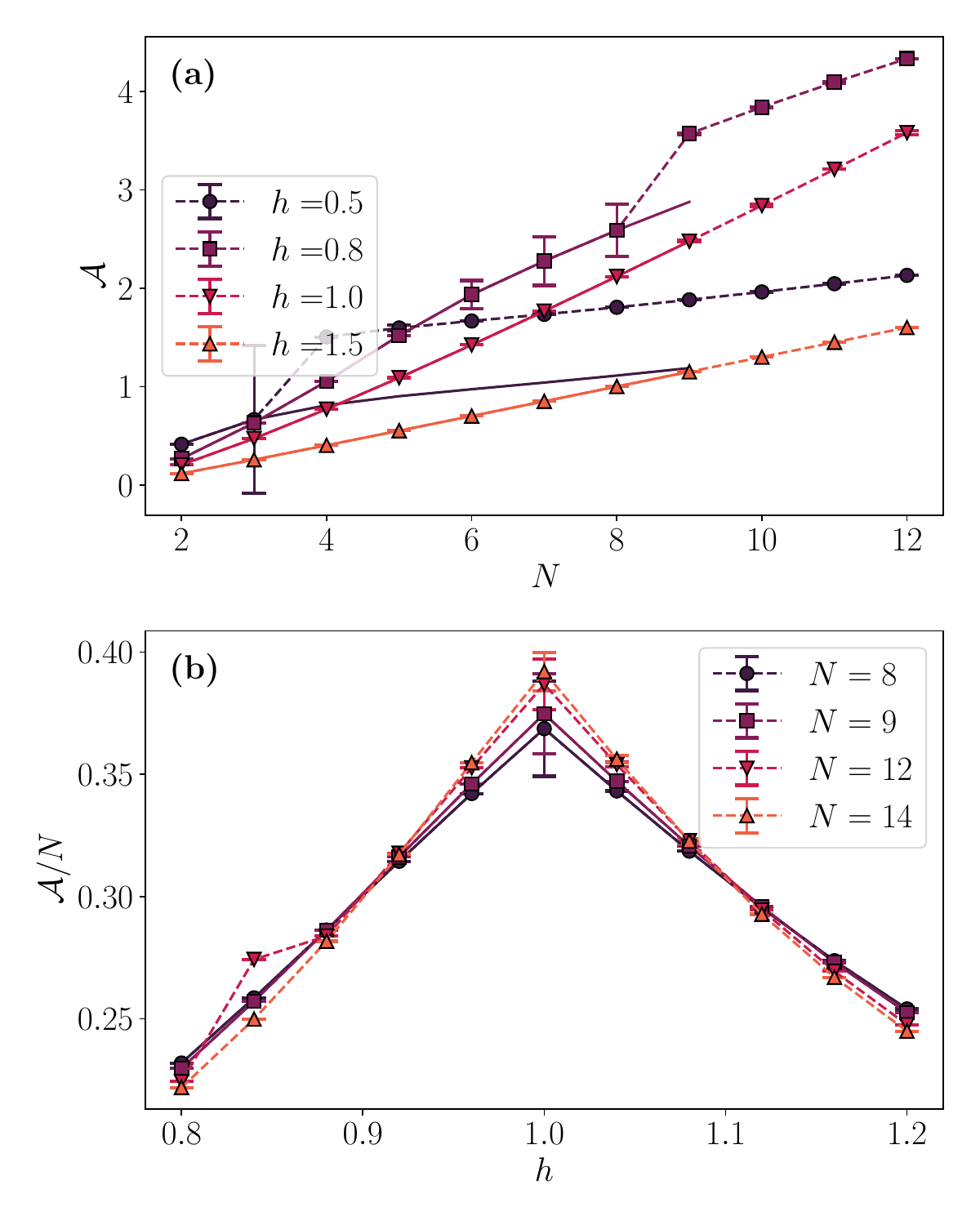}
    \caption{ BMSA obtained via Monte Carlo minimization of participation entropy as explained in Sec.~\ref{sec:min_part_entropy}. \textbf{(a)} Scaling of BMSA for increasing system sizes in the ground state of the quantum Ising chain with open boundary conditions. Solid lines correspond to the exact calculation for $N$ up to 9 obtained via the branch and bound method as explained in Sec.~\ref{sec:branch_and_bound}.
    \textbf{(b)} BMSA density in transverse field Ising model with periodic boundary conditions, for different system sizes. Solid lines correspond to the exact calculation for $N \in \{8, 9\}$, obtained via branch and bound method. }
    \label{fig:asymmetry_ising}
\end{figure}

We checked the performance of the algorithm in the quantum Ising chain described by the following Hamiltonian:
\begin{equation}
    H_{\text{Ising}} = -\sum_{i=1}^{N-1} \sigma_i^z \sigma_{i+1}^z - h \sum_i \sigma^x_i,
    \label{eq:1dising}
\end{equation}
with $\sigma^{z,x}$ spin-$1/2$ Puali matrices. Nonstabilizerness of the ground-state of the Ising Hamiltonian has been studied in Ref. \cite{oliviero2022magic} by means of free fermion techniques and in Refs. \cite{haug2022quantifying,tarabunga2024nonstabilizerness} using matrix product states. It has been shown that it peaks at the criticality $h_c=1$. In Fig. \ref{fig:asymmetry_ising}(a) we present the scaling of the BMSA as a function of system size $N$ for the ground state of $H_{\text{Ising}}$ for different transverse fields $h$ with open boundary conditions, showing extensive behavior in both phases. We observed that the algorithm occasionally becomes trapped in local minima within the ferromagnetic phase.  This is evident from the curves of Fig.~\ref{fig:asymmetry_ising}(a), showing disagreement with the numerically-exact data for $h=0.5$ and $N\geq 4$, and for $h=0.8$, $N=9$. To mitigate this issue, we initialized the algorithm with an initial guess of $C_\mathrm{ref}=H_1 CNOT_{1,2} CNOT_{2,3} \dots CNOT_{N-1,N}$. This choice corresponds to the Clifford gate that minimizes the participation entropy at $h=0$. In Fig. \ref{fig:asymmetry_ising}(b) we show the density of the BMSA as a function of the transverse field $h$ with periodic boundary conditions, focusing around the critical point $h_c=1$. For systems with $N \in \{8,9\}$, we observe perfect agreement between the numerical data and the exact values computed using the branch and bound method in Sec. \ref{sec:branch_and_bound}. Furthermore, we give a glimpse of the applicability of our technique for larger systems with $N \in \{12,14\}$, showcasing the potential of our approach to handle systems beyond the reach of exact numerical methods.

\subsection{Analytic computations for upper and lower bounds}\label{sec:bounds}
The previous methods allow us to numerically compute the BMSA with some versatility. However, in some cases, it is possible to arrive at analytic results in the large-$N$ limit. As an interesting example, we consider the so-called $W$-state
\begin{equation}
\ket{W} = \frac{1}{\sqrt{N}} \qty(\ket{100\dots0}+\ket{010\dots0}+\dots+\ket{000\dots1}).
\end{equation}
Our strategy is to exploit previously derived upper bounds, while using a specific stabilizer $\ket{s}$ to get lower bounds which are easily computable. While we illustrate this method for a simple state, we expect that similar strategies could give analytic large-$N$ results in special families of states, such as matrix-product states~\cite{haug2022quantifying}.

To be concrete, by definition the participation entropy in any stabilizer basis provides an upper bound to the BMSA. This can be easily computed for the $W$-state in the computational basis, which is given by $\log(N)$. We then use the lower bound provided by the SRE in Eq. \eqref{eq:relation_sre}. The SRE-2 for the W-state was computed in \cite{odavic2023complexity}, and is given by $M_2(\ket{W})=3\log(N)-\log(7N-6)$. We thus have
\begin{equation}
   \frac{3}{2} \log(N)- \frac{1}{2} \log(7N-6) \leq \mathcal{A}_\alpha(\ket{W}) \leq \log(N), 
\end{equation}
for $\alpha \leq 2$. For large $N$, both of the bounds collapse to $\log(N)$, yielding
\begin{equation}
\mathcal{A}_\alpha(\ket{W}) \sim \log(N)\,,
\end{equation}
for $\alpha \leq 2$ and large $N$.

Additionally, with a similar argument, we can show that the BMSA scales linearly in $N$ in the large $N$ limit for product states. Indeed, since both the participation entropy and the SRE are additive for product states, the BMSA is bounded above and below by extensive quantities. This shows that the BMSA is an extensive quantity for product states. 

\section{Outlook}
\label{sec:outlook}
We have introduced a measure of nonstabilizerness, the BMSA. We have shown that it is a strong monotone and provided evidence that it is additive, up to corrections that vanish in the limit of large numbers of qubits. We have shown that the BMSA can be computed exactly up to $N=9$ qubits, and developed a Monte Carlo approach to compute it for larger $N$. The Monte Carlo method appears to match well the numerically-exact results for the available system sizes. We have derived several inequalities between the BMSA and other nonstabilizerness monotones, showing how these inequalities may allow one to infer the leading behavior of the BMSA in the large-$N$ limit.

Our work opens several questions. Most prominently, it would be interesting to devise numerically exact or controlled approximate methods to compute the BMSA in special classes of many-body states, such as MPSs, for large number of qubits. As mentioned, we have shown that, in some cases, it is possible to compute analytically lower and upper bounds, from which one can infer the exact leading behavior of the BMSA in the large-$N$ limit. While we have detailed such calculation only for a simple example, the $W$-state, we expect that a similar approach could work to compute, possibly only numerically, the large-$N$ limit of the BMSA more generally. Finally, another important question to address is the strong monotonicity of the R\'enyi BMSA for $1<\alpha<2$, and $\alpha<1, \alpha \neq 1/2$.

Going further, our work contributes to the recent literature aiming at finding efficiently computable monotones in the nonstabilizerness resource theory. As we have stressed the importance of additivity and strong monotonicity in the context of many-body physics, we hope that our work will motivate further research to find additional monotones with these properties.

\section*{Acknowledgments}
We thank Marcello Dalmonte for inspiring discussions and collaboration at an early stage of this work. P.S.T. acknowledges funding by the Deutsche Forschungsgemeinschaft (DFG, German Research Foundation) under Germany’s Excellence Strategy – EXC-2111 – 390814868.
This work was co-funded by the European Union (ERC, QUANTHEM, 101114881). E.T. acknowledges support from ERC under grant agreement n.101053159 (RAVE).
Views and opinions expressed are however those of the author(s) only and do not necessarily reflect those of the European Union or the European Research Council Executive Agency. Neither the European Union nor the granting authority can be held responsible for them.

\appendix

\section{Properties of the BMSA}
\label{sec:properties_BMSA}

In this appendix, we provide additional details on the BMSA and prove that it is a strong monotone. 

\subsection{Preliminary properties}

We start by deriving Eq.~\eqref{eq:rho_sym}. Using the same notation as in the main text, the symmetrized state takes the form
\begin{equation}
    \begin{split}
        \mathcal{G}_G(\rho) &= \frac{1}{2^k} \sum_{g \in G} U_g \rho U_g^{\dagger} \\
        &= \frac{1}{2^k} \sum_{g \in G}  U_g \left( \frac{1}{2^N} \sum_{P\in \mathcal{P}_N} \Tr[\rho P] P \right) U_g^{\dagger} \\
        &= \frac{1}{2^{N+k}} \sum_{g \in G}    \sum_{P\in \mathcal{P}_N} \Tr[\rho P] U_g P  U_g^{\dagger} \\
        &= \frac{1}{2^N}  \sum_{P\in \mathcal{P}_N} \Tr[\rho P] \left( \frac{1}{2^k} \sum_{g \in G} U_g P  U_g^{\dagger} \right). \\ 
    \end{split}
\end{equation}
It is easy to see that the term inside the bracket is $P$ if and only if $P$ commutes with all $g \in G$, and 0 otherwise. Therefore, denoting by $G^{\perp}$ the group of Pauli strings we obtain Eq.~\eqref{eq:rho_sym}.

Next, we show Eq.~\eqref{eq:second_inequality}.
To this end, we follow the approach of Ref.~\cite{howard2017application} for the bound on robustness of magic and SRE. First, let us assume that we have the decomposition $\ket{\psi}=\sum_i x_i d_i\ket{s}$ with minimal asymmetry as in Eq.~\eqref{eq:asymmetryProgram}. 
Then, let us consider
\begin{gather}
    e^{\frac{1}{2}M_{1/2}(\ket{\psi})}=2^{-N}\sum_{P\in\mathcal{P}_N}\vert \bra{\psi}P\ket{\psi}\vert\nonumber\\
    =2^{-N}\sum_{P\in\mathcal{P}_N}\vert \sum_{ij} x_i^\ast x_j\bra{s}d_i^\dagger P d_j\ket{s}\vert\nonumber\\
    \leq 2^{-N}\sum_{ij}\vert x_i^\ast\vert \vert x_j\vert \sum_{P\in\mathcal{P}_N}\vert\bra{s}d_i^\dagger P d_j\ket{s}\vert
\end{gather}
where we used the triangle inequality.
Now, we note that as we sum over all Pauli group, the result remains unchanged for any $P\rightarrow U_\text{C} P U_\text{C}^\dagger$ with any Clifford $U_\text{C}$. Now, we choose $U_\text{C}$ such that $U_\text{C} d_j \ket{s}=\ket{0}^{\otimes N}$. Note that $U_C$ maps the stabilizer basis to the computational basis, so that $U_\text{C} d_i \ket{s}$ is a computaional basis state also for $i \neq j$. Then, we have
\begin{gather}
    2^{-N}\sum_{ij}\vert x_i^\ast\vert \vert x_j\vert \sum_{P\in\mathcal{P}_N}\vert\bra{s}d_i^\dagger U_\text{C}^\dagger P\ket{0}^{\otimes N}\vert  \\
    =2^{-N}\sum_{ij}\vert x_i^\ast\vert \vert x_j\vert \sum_{\sigma} 2^{N}\vert\bra{s}d_i^\dagger U_\text{C}^\dagger \ket{\sigma}\vert
\end{gather}
where $\ket{\sigma}$ denotes a computational basis state. Since $U_C d_i \ket{s}$ is a computational basis state, it follows that $\sum_{\sigma}\vert\bra{s}d_i^\dagger U_\text{C}^\dagger \ket{\sigma}\vert=1$. Thus, we have
\begin{align}
    e^{\frac{1}{2}M_{1/2}(\ket{\psi})} &\leq \left(\sum_{i} \abs{x_i}\right)^2=  e^{\mathcal{A}_{1/2}(\ket{\psi})}
\end{align}

\subsection{Strong monotonicity}
We now prove strong monotonicity. We begin by proving some preliminary results.

\begin{lem}\label{lem:measurements}
    The BMSA is non-increasing, on averaged, under Pauli measurements
\end{lem}
\begin{proof}
Under the measurement of a Pauli operator $Q$, the state $\rho = \ket{\psi} \bra{\psi}$ is projected to $\rho_{\lambda} = V_{\lambda} \rho V_{\lambda} / p_{\lambda}$, where $V_{\lambda} = (I + \lambda Q)/2$ and $p_{\lambda} = \Tr \rho V_{\lambda}$. Let $G(\ket{s})$ be the stabilizer group such that $\mathcal{A}(\ket{\psi}) = A_{G(\ket{s})}(\rho) $. There are two cases:

\begin{enumerate}
\item $Q \in G(\ket{s})$ \\
In this case, measurement of $Q$ is a $G(\ket{s})$-invariant operation. Thus, we have
 \begin{equation}
        \begin{split}
             \mathcal{A}(\ket{\psi}) &= A_{G(\ket{S})}(\rho) \\ 
             &\geq \sum_{\lambda} p_\lambda A_{G(\ket{S})}(\rho_\lambda) \\
             &\geq \sum_{\lambda} p_\lambda \mathcal{A}(\ket{\psi_\lambda}).
        \end{split}
        \end{equation}
        The second line holds since $A_{G(\ket{s})}(\rho)$ is a $G(\ket{s})$-frame strong monotone \cite{vaccaro2008tradeoff}.
       
        \item $Q \notin G(\ket{s})$ \\
         Let $\sigma_{\lambda} = V_{\lambda} \mathcal{G}_{G(\ket{s})}(\rho) V_{\lambda} / q_{\lambda}$, where $q_{\lambda} = \Tr \mathcal{G}_{G(\ket{s})}(\rho) V_{\lambda}$. Let the Pauli measurement maps $\ket{s}$ to $\ket{s^{\prime}}$. Then, we claim that 
         $\sigma_\lambda$ is $G(\ket{s^{\prime}})$-invariant.
        
        To see this, let $Q, g_2, \cdots, g_N $ be the generators of $G(\ket{s^{\prime}})$. Note that we can choose these generators such that $g_2, \cdots, g_N \in G(\ket{s})$ (this follows from the update rule for Pauli measurement, see \cite{aaronson2004improved}). To prove the claim, it is sufficient to show that $\sigma_\lambda$ is invariant under the generators of $G(\ket{s^{\prime}})$. The invariance under $Q$ is obvious since $\sigma_\lambda$ is stabilized by $Q$.
        Now, for $j \in \{2, \cdots, N \}$, we have
        \begin{align}
        g_j \sigma_{\lambda} g_j &= \frac{g_j V_{\lambda} \mathcal{G}_{G(\ket{s})}(\rho) V_{\lambda} g_j}{q_\lambda}\nonumber\\
        & =  \frac{V_{\lambda} g_j \mathcal{G}_{G(\ket{s})}(\rho) g_j  V_{\lambda}}{q_\lambda}  \nonumber\\
        & =\frac{V_{\lambda} \mathcal{G}_{G(\ket{s})}(\rho)   V_{\lambda}}{q_\lambda}  = \sigma_{\lambda}.  
        \end{align}
  
        The second equality holds since $g_j$ commutes with $Q$, while the third equation follows since $g_j \in G(\ket{s})$ and $\mathcal{G}_{G(\ket{s})}(\rho)$ is $G(\ket{s})$-invariant. This proves our claim.
        
        Therefore,
        \begin{equation} \label{eq:a4}
        \begin{split}
             \mathcal{A}(\ket{\psi}) &= S(\rho \Vert \mathcal{G}_{G(\ket{s})}(\rho)) \\ 
             &\geq \sum_{\lambda} p_\lambda S(\rho_\lambda \Vert \sigma_\lambda) \\
             &\geq \sum_{\lambda} p_i A_{G(\ket{S^{\prime}})}(\rho_\lambda ) \\
             &\geq \sum_{\lambda} p_\lambda \mathcal{A}(\ket{\psi_\lambda}) .
        \end{split}
        \end{equation}
        The second line was proven in Ref.~\cite{vedral1998entanglement} (and also used to prove the strong monotonicity of the relative entropy of magic, see Theorem 7 in Ref.~\cite{veitch2014resource}).
        The third line follows from the fact that $\mathcal{G}_{G(\ket{s^{\prime}})}(\rho_i)$ is the nearest  $G(\ket{s^{\prime}})$-invariant state to $\rho_i$~\cite{gour2009measuring}. 
\end{enumerate}
\end{proof}

Next, we prove
\begin{lem}\label{lem:appending}
    The BMSA is invariant when discarding stabilizer qubits, namely
    \begin{equation}
        \mathcal{A}(\ket{u}\otimes \ket{\psi})=\mathcal{A}(\ket{\psi})\,,
    \end{equation}
    for any stabilizer state $\ket{u}$.
\end{lem}
\begin{proof}
Let us denote by $\mathcal{H}_1$, $\mathcal{H}_2$ the Hilbert spaces corresponding to $\ket{u}$ and $\ket{\psi}$, respectively, so that
\begin{equation}
    \ket{\psi}\bra{\psi}=\Tr_1\left[\ket{u}\bra{u}\otimes \ket{\psi}\bra{\psi}\right]\,.
\end{equation}
We first prove that
\begin{equation}\label{eq:ineq_1}
    \mathcal{A}(\ket{\psi}\bra{\psi})\leq    \mathcal{A}(\ket{u}\bra{u}\otimes \ket{\psi}\bra{\psi})
\end{equation}
To see this, let $\ket{s}\in \mathcal{H}_1\otimes \mathcal{H}_2$ be a stabilizer state and $G(\ket{s})$ the stabilizer group such that $\mathcal{A}(\ket{u}\otimes\ket{\psi}) = A_{G(\ket{s})}(\rho)$, where $\rho = \ket{u}\bra{u}\otimes \ket{\psi}\bra{\psi}$. We again consider two cases:
    \begin{enumerate}
        \item $\Tr_1\left[\ket{s} \bra{s}\right]$ is a pure state. \\
        It was shown in \cite{vaccaro2008tradeoff} that $A_{G}$ is non-increasing under partial trace when $G$ is factorized, from which the result follows immediately.
        \item $\Tr_1\left[\ket{s} \bra{s}\right]$ is a mixed state. \\
        Let $\Tr_1\left[\ket{s} \bra{s}\right]= \rho_s$. Let $G(\rho_s)$ be the stabilizer group of $\rho_s$, which is generated by $g_1, \cdots,g_r$ where $r<N-1$. Using 
\begin{equation} \label{eq:twirled_stab}
\mathcal{G}_{G(| s \rangle)}(\rho) = \frac{1}{2^N} \sum_{P \in G(| s \rangle)} \Tr[\rho P] P,
\end{equation}
we have 
        \begin{equation} \label{eq:trace_twirled}
            \Tr_1 \mathcal{G}_{G(| s \rangle)}(\rho) = \frac{1}{2^N} \sum_{P \in G(\rho_s)} \Tr[\rho P] P.
        \end{equation}
        We can add generators $g_{r+1}, \cdots, g_{N-1}$, all mutually commute with $g_j, j \leq r$, to $G(\rho_s)$ to make it a stabilizer group of a pure stabilizer state, say $G(\ket{s'})$. It follows from Eq. \eqref{eq:trace_twirled} that $\Tr_1 \mathcal{G}_{G(| s \rangle)}(\rho)$ is also invariant under $g_j, j > r$. Thus, $\Tr_1 \mathcal{G}_{G(| s \rangle)}(\rho)$ is $G(\ket{s'})$-invariant. The proof then proceeds similarly as in Eq. \eqref{eq:a4}.
        \end{enumerate}
On the other hand, by sub-additivity, we have
\begin{align}
\mathcal{A}(\ket{u}\bra{u}\otimes\ket{\psi}\bra{\psi})&\leq\mathcal{A}(\ket{u}\bra{u})+\mathcal{A}(\ket{\psi}\bra{\psi})\nonumber\\
&=\mathcal{A}(\ket{\psi}\bra{\psi})\,.\label{eq:ineq_2}
\end{align}
Finally, the statement of the Lemma follows combining Eqs.~\eqref{eq:ineq_1} and ~\eqref{eq:ineq_2}

\end{proof}
With the help of Lemmas~\ref{lem:measurements} and ~\ref{lem:appending} we can easily prove that  $\mathcal{A}$ is a stabilizerness monotone.

First, it is immediate to see that $\mathcal{A}$ is invariant under unitary Clifford operations. Next, Lemma~\ref{lem:appending} implies immediately that $\mathcal{A}$ is invariant if we discard/append ancillary qubits initialized in pure stabilizer states. Now, let $\mathcal{E}$ be a stabilizer protocol, consisting in a local measurement followed by conditioned Clifford operations that are either unitary operators or the removal/addition of an ancilla. Suppose that $\mathcal{E}(\ket{\psi}\bra{\psi})=\ket{\phi}\bra{\phi}$. We denote by $M_k$ the Kraus operator encoding the update of the system state after obtaining the measurement outcome $k$, so that
\begin{equation}
    \mathcal{E}(\cdot)= \sum_k M_k(\cdot) M^\dagger_k\,.
\end{equation}
Then, setting $p_k$ as the probability of obtaining the measurement outcome $k$ Lemma~\ref{lem:measurements} implies
\begin{equation}
    \sum_k p_k \mathcal{A}(M_k \rho_k M_k^\dagger)\leq \mathcal{A}(\ket{\psi}\bra{\psi})
\end{equation}
On the other hand, since $\mathcal{E}(\ket{\psi}\bra{\psi})$ is pure, $M_k \rho_k M_k^\dagger=\ket{\phi}\bra{\phi}$ for all $k$, so that 
$\sum_k p_k \mathcal{A}(M_k \rho_k M_k^\dagger)= \mathcal{A}(\ket{\phi}\bra{\phi})$ which proves monotonicity. The same argument applies for arbitrary sequences of measurements followed by Clifford operations.

Finally, the strong monotonicity is precisely the statement of Lemma~\ref{lem:measurements}.

\subsection{Proof of Eq.~\eqref{eq:to_prove_1}}
\label{sec:proof_inequality}
In this Appendix we prove Eq.~\eqref{eq:to_prove_1}. To this end, we introduce the Gowers distance~\cite{arunachalam2024tolerant}
    \begin{equation}
        \text{Gowers}(\ket{\psi},3)^8=2^{-N}\sum_{P\in \mathcal{P}_N} \bra{\psi}P\ket{\psi}^4\,.
    \end{equation}
    The Gowers distance is related to the SRE via
    \begin{equation}
        \text{Gowers}(\ket{\psi},3)^8=\exp(-M_2(\ket{\psi}))\,.
    \end{equation}
    Then, Ref.~\cite{arunachalam2024notepolynomialtimetoleranttesting} showed that there exist a constant $C>1$ such that
    \begin{equation}
        F_\text{STAB}(\ket{\psi})\geq \text{Gowers}(\ket{\psi},3)^{8C}\,.
    \end{equation}
    This directly implies 
        \begin{equation}
        D_\text{min}(\ket{\psi})\leq C M_2(\ket{\psi})
    \end{equation}
    and together with Eq.~\eqref{eq:upper_dmin} we get
        \begin{equation}
       \mathcal{A}_2(\ket{\psi})\leq 2C M_2(\ket{\psi})\,.
    \end{equation}

\section{Rényi BMSA}
\label{sec:renyi_BMSA}

In this appendix we discuss the properties of the Rényi-$2$ BMSA. It has the simple form
\begin{equation} \label{eq:A2}
    \mathcal{A}_{2} ( \ket{\psi}) = \min_{|S \rangle \in \text{PSTAB}_N} -\log \sum_{P\in G(\ket{S})} \frac{\Tr[\rho P]^2}{2^N} .
\end{equation}
As anticipated, this is not a strong monotone, as can be shown using the construction in Ref.~\cite{haug2023stabilizer}. We can construct a strong monotone using the linear version:
\begin{equation} \label{eq:A2_lin}
    \mathcal{A}_{2}^{\mathrm{lin}} ( \ket{\psi}) = \min_{|S \rangle \in \text{PSTAB}_N} 1- \sum_{P\in G(\ket{S})} \frac{\Tr[\rho P]^2}{2^N} .
\end{equation}
The proof is as follows.
\begin{lem}\label{lem:measurements_Alin}
$\mathcal{A}_{2}^{\mathrm{lin}} ( \ket{\psi})$ is non-increasing, on averaged, under Pauli measurements
\end{lem}
\begin{proof}
Under the measurement of a Pauli operator $Q$, the state $\rho = \ket{\psi} \bra{\psi}$ is projected to $\rho_{\lambda} = V_{\lambda} \rho V_{\lambda} / p_{\lambda}$ for $\lambda=\pm 1$, where $V_{\lambda} = (I + \lambda Q)/2$ and $p_{\lambda} = \Tr \rho V_{\lambda}$. Let $G(\ket{s})$ be the stabilizer group such that $\mathcal{A}(\ket{\psi}) = A_{G(\ket{s})}(\rho) $. We will first show that, for any Pauli operator $P$ that commutes with $Q$, the following holds:
    \begin{equation} \label{eq:lemma2}
        \sum_{\lambda} p_\lambda  \Tr[\rho_\lambda P]^2  \geq  \Tr[\rho P]^2 .
    \end{equation}
    Indeed, we have
    \begin{equation}
    \begin{split}
        \Tr[\rho_\lambda P] &= \frac{1}{p_\lambda}  \Tr[V_\lambda \rho V_\lambda P] \\
        &= \frac{1}{p_\lambda}  \Tr[V_\lambda \rho P V_\lambda] \\
        &= \frac{1}{p_\lambda}  \Tr[\rho P V_\lambda] \\
        &= \frac{1}{2 p_\lambda}  \Tr[\rho P] + \lambda \Tr[\rho PQ].  \\
    \end{split}
    \end{equation}
    We also have $p_\lambda = (1 + \lambda \Tr[\rho Q])/2$. Plugging these to the l.h.s. of Eq. \eqref{eq:lemma2}, we get
    \begin{equation}
    \begin{split}
        \sum_{\lambda} p_\lambda  \Tr[\rho_\lambda P]^2 
        = \left(\Tr[\rho P]^2 + \Tr[\rho PQ]^2 - 
        \right.
        \\
        \left.
        2\Tr[\rho P]\Tr[\rho PQ]\Tr[\rho Q] \vphantom{\Tr[\rho PQ]^2 -}\right)/\left(1 - \Tr[\rho Q]^2\right).     \\
    \end{split}
    \end{equation}
    Therefore, 
    \begin{equation}
    \begin{split}
        \sum_{\lambda} p_\lambda  \Tr[\rho_\lambda P]^2 -  \Tr[\rho P]^2 
        &= \frac{(\Tr[\rho P] \Tr[\rho Q] - \Tr[\rho PQ])^2}{(1 - \Tr[\rho Q]^2)} \\
        & \geq 0.  \\ 
    \end{split}
    \end{equation}

    There are two cases:
    \begin{enumerate}
        \item $Q \in G(\ket{s})$ \\
        Summing Eq. \eqref{eq:lemma2} over all $P \in G(\ket{s})$ directly leads to the desired result.
        \item $Q \notin G(\ket{s})$ \\
        Let the Pauli measurement maps $\ket{s}$ to $\ket{s^{\prime}}$.  Let $Q, g_2, \cdots, g_N $ be the generators of $G(\ket{s^{\prime}})$, such that $g_j \in G(\ket{s})$. Let $T$ be the group generated by $g_j, j\in \{2, \cdots, N\}$. We have
        \begin{equation} \label{eq:24}
            \begin{split}
                \sum_{\lambda} p_\lambda \sum_{P\in G(\ket{s'})} \Tr[\rho_\lambda P]^2 &= 2\sum_{\lambda} p_\lambda \sum_{P\in T} \Tr[\rho_\lambda P]^2 \\
                &\geq 2\sum_{P\in T} \Tr[\rho P]^2. \\
            \end{split}
        \end{equation}
        The last line follows by summing Eq. \eqref{eq:lemma2} over all $P \in T$. Now, let $T^{\perp}$ be the group of Pauli operators that commute with $T$. We have
        \begin{equation} \label{eq:25}
            \begin{split}
                \sum_{P\in T} \Tr[\rho P]^2 &= \frac{|T|}{2^N} \sum_{P\in T^{\perp}} \Tr[\rho P]^2 \\
                &\geq \frac{1}{2} \sum_{P\in G(\ket{s'})} \Tr[\rho P]^2. \\
            \end{split}
        \end{equation}
        The first line is the Theorem 3.1 in \cite{grewal2024improved}, while the last line comes from the fact that $G(\ket{s'}) \subset T^{\perp}$. Finally, the result follows by combining Eq. \eqref{eq:24} and Eq. \eqref{eq:25}.
        \end{enumerate}
\end{proof}

Next, we prove
\begin{lem}\label{lem:appending_alin}
$\mathcal{A}_{2}^{\mathrm{lin}} ( \ket{\psi})$ is invariant when discarding stabilizer qubits, namely
    \begin{equation}
    \mathcal{A}_{2}^{\mathrm{lin}}(\ket{u}\otimes \ket{\psi})=\mathcal{A}_{2}^{\mathrm{lin}}(\ket{\psi})\,,
    \end{equation}
    for any stabilizer state $\ket{u}$.
\end{lem}
\begin{proof}
    Note that $ \Tr(\mathcal{G}_{G(\ket{s})}(\rho)^2) = \Tr(\rho \mathcal{G}_{G(\ket{s})}(\rho))$, which is equivalent to the fidelity between $\rho$ and $\mathcal{G}_{G(\ket{s})}(\rho)$. The proof then proceeds similarly as in Lemma \ref{lem:appending}, using the fact that the fidelity is also monotonic under partial trace.
\end{proof}

Note that, using the strong monotonicity of $\mathcal{A}_{2}^{\mathrm{lin}} ( \ket{\psi})$, one can show that $\mathcal{A}_{2}( \ket{\psi})$ is a genuine monotone (although not strong), using the same arguments as in Ref.~\cite{leone2024stabilizer}.

Finally, we comment on efficient measurement schemes for $\mathcal{A}_{2}( \ket{\psi})$ and $\mathcal{A}_{2}^{\mathrm{lin}} ( \ket{\psi})$. 
This can be achieved by using Bell difference sampling~\cite{montanaro2017learning} as a measurement primitive, where we outline two different methods.

In the first approach, one can adapt the algorithm of  Ref.~\cite{grewal2024improved} for measuring the min-relative entropy of magic, which relies on Bell difference sampling, to also estimate $\mathcal{A}_{2}( \ket{\psi})$.

In the second approach, we define the convoluted BMSA
\begin{equation} \label{eq:A2_conv}
    \mathcal{A}_{2}^\text{conv} ( \ket{\psi}) = \min_{|S \rangle \in \text{PSTAB}_N} -\log \sum_{P_{\boldsymbol{a}}\in G(\ket{S})} \sum_{\boldsymbol{a}} \Xi(\boldsymbol{a}) \Xi(\boldsymbol{a} \oplus \boldsymbol{b})\,,
\end{equation}
as another magic monotone which is closely related to the R\'enyi-2 BMSA. The convoluted BMSA is constructed by replacing the characteristic function $\Xi(\boldsymbol{a})=\langle \psi | P_{\boldsymbol{a}} | \psi \rangle^{2}/2^N$ of the BMSA by its convolution $Q(\boldsymbol{b})=\sum_{\boldsymbol{a}} \Xi(\boldsymbol{a}) \Xi(\boldsymbol{a} \oplus \boldsymbol{b})=\langle \psi | P_{\boldsymbol{a}} | \psi^\ast \rangle^{2}/2^N$ where there is a bijective mapping between $\boldsymbol{a}\in\{0,1\}^{2N}$ and Pauli $P_{\boldsymbol{a}}$~\cite{montanaro2017learning,gross2021schur}. One can show that $\mathcal{A}_{2}^\text{conv} ( \ket{\psi})$ has the same properties as $\mathcal{A}_{2} ( \ket{\psi})$. Conveniently, the convolution $Q(\boldsymbol{b})$ can be efficiently sampled via Bell difference sampling~\cite{montanaro2017learning}, allowing one to estimate $ \mathcal{A}_{2}^\text{conv} ( \ket{\psi})$ directly from the sampling outcomes.

\section{Branch and bound}
In this section, we provide further details on the method explained in Sec.~\eqref{sec:branch_and_bound}.

\subsection{Derivation of lower bound in Eq. \eqref{eq:bound_r}} \label{sec:deriv_bound}
We first derive the lower bound given in Eq.~\eqref{eq:bound_r}. We will show this for $\alpha=2$ and note that similar steps can be carried out for any integer $\alpha>2$. By expanding the power and after some algebra, we can show
  \begin{equation}
  \begin{split}
\sum_{Q_\mathrm{d},t}p^2_{Q',c,R}(Q_\mathrm{d},t) = \sum_t \sum_{x_1=0}^{2^k-1} \sum_{x_2=0}^{2^k-1} \sum_{x_3=0}^{2^k-1}   \left[f_{x_1,x_2,x_3} \vphantom{c_{Rx_{123}+t}^\dagger} \right. 
  \\ 
  \left.
  c_{Rx_1+t} c_{Rx_2+t}^\dagger c_{Rx_3+t} c_{Rx_{123}+t}^\dagger \right],
  \end{split}
  \end{equation}
    where $f_{x_1,x_2,x_3}=(-1)^{x_1^\top Q' x_1 + x_2^\top Q' x_2 + x_3^\top Q' x_3 + x_{123}^\top Q' x_{123}} \\(-i)^{c^\top x_1} i^{c^\top x_2} (-i)^{c^\top x_3} i^{c^\top x_{123}}$  and $x_{123} = x_1 + x_2 + x_3$. Then, by triangle inequality, we have
    \begin{equation}
    \begin{split}
\sum_{Q_\mathrm{d},t}p^2_{Q',c,R}(Q_\mathrm{d},t) &\leq \sum_t \sum_{x_1=0}^{2^k-1} \sum_{x_2=0}^{2^k-1} \sum_{x_3=0}^{2^k-1}    \left[\abs{c_{Rx_1+t}} \abs{c_{Rx_2+t}^\dagger} \right. 
\\
\left.
  \abs{c_{Rx_3+t}} \abs{c_{Rx_{123}+t}^\dagger}\right] \\
        &= \sum_{Q_\mathrm{d},t} q^2_{R}(Q_\mathrm{d},t),
        \end{split}
    \end{equation}
which immediately implies Eq. \eqref{eq:bound_r}.

\subsection{Simplification for wavefunction with real amplitudes} \label{sec:real_states}

Next, we show that the minimization in Eq.~\eqref{eq:bound_r} can be simplified if the wavefunction has real amplitudes. Some physically relevant applications are groundstates of Hamiltonians with time-reversal symmetry. Indeed, if $\ket{\psi}$ has real amplitudes, it can be shown that the minimum in Eq. \eqref{eq:bmsa2} is achieved by $c=\mathbf{0}^{k}$, for any $\alpha \geq 0$. To see this, it suffices to show that 
  \begin{equation}  \label{eq:real_amp}
  S_\alpha(\{p_{Q',c=\mathbf{0}^{k},R}(Q_\mathrm{d},t) \}_{Q_\mathrm{d},t}) \leq S_\alpha(\{p_{Q',c,R}(Q_\mathrm{d},t) \}_{Q_\mathrm{d},t})
  \end{equation}
  for any $Q',c$. We consider three cases:
  \begin{enumerate}
      \item $\alpha>1$. For a stabilizer state 
      \begin{equation}
          \ket{\phi} = \frac{1}{2^{k/2}}\sum_{x=0}^{2^k-1} (-1)^{x^\top Q' x} (-1)^{Q_\mathrm{d}^\top x} i^{c^\top x} \ket{Rx+t}\,,
      \end{equation}
     suppose $\braket{\phi}{\psi} = \beta + i \gamma$, where $\beta, \gamma \in \mathbb{R} $. If $c\neq \mathbf{0}^{k}$, we can construct stabilizer states
      \begin{align*}
      \ket{\phi'} &= \frac{1}{2^{k/2}}\sum_{x=0}^{2^k-1} (-1)^{x^\top Q' x} (-1)^{(Q_\mathrm{d}^\top + c^\top) x} i^{c^\top x}\ket{Rx+t}, \\
    \ket{\phi_+} &= \frac{1}{2^{k/2}}\sum_{x=0}^{2^k-1} (-1)^{x^\top Q' x} (-1)^{Q_\mathrm{d}^\top x}\ket{Rx+t}, \\
    \ket{\phi_-} &= \frac{1}{2^{k/2}}\sum_{x=0}^{2^k-1} (-1)^{x^\top Q' x} (-1)^{(Q_\mathrm{d}^\top + c^\top) x} \ket{Rx+t},
  \end{align*}
  such that  $\braket{\phi'}{\psi} = \beta - i \gamma$, $\braket{\phi_+}{\psi} = \beta +  \gamma$, and $\braket{\phi_-}{\psi} = \beta - \gamma$. We claim that
  \begin{equation} \label{eq:real_alpha>1}
      \abs{\braket{\phi'}{\psi}}^{2\alpha} + \abs{\braket{\phi'}{\psi}}^{2\alpha} \leq \abs{\braket{\phi_+}{\psi}}^{2\alpha} + \abs{\braket{\phi_-}{\psi}}^{2\alpha},
  \end{equation}
  for any $\alpha > 1$, which can be written as
  \begin{align*}
      2(\beta^2 + \gamma^2)^\alpha &\leq (\beta + \gamma)^{2\alpha} + (\beta - \gamma)^{2\alpha} \\
      1 &\leq  (1+r)^{\alpha} + (1-r)^{\alpha},
  \end{align*}
  where $r=\beta \gamma/(\beta^2 + \gamma^2)$. The latter inequality can be directly shown through H\"older's inequality. Finally, the inequality in Eq. \eqref{eq:real_amp} is obtained by summing Eq. \eqref{eq:real_alpha>1} over $Q_d,t$ and taking the logarithm.
  \item $0<\alpha<1$. This case is similar as the previous case, except that we need to show the reversed inequality
    \begin{equation}
      \abs{\braket{\phi'}{\psi}}^{2\alpha} + \abs{\braket{\phi'}{\psi}}^{2\alpha} \geq \abs{\braket{\phi_+}{\psi}}^{2\alpha} + \abs{\braket{\phi_-}{\psi}}^{2\alpha},  \end{equation}  
      for $\alpha<1$. This can again be shown through H\"older's inequality.
    \item $\alpha=1$ and  $\alpha=0$. This case follows by continuity.
    \end{enumerate}

Finally,  if $\ket{\psi}$ has positive real amplitudes, it is easy to see that the lower bound in Eq. \eqref{eq:bound_r} can be achieved by $Q' = \mathbf{0}^{k \times k}$ and $c=\mathbf{0}^{k}$.

\section{Additional numerics}
In this Appendix, we show numerical evaluation of BMSA and other measures of nonstabilizerness, in particular $\text{LR}$, $r_\mathcal{M}$ and $D_\text{min}$. We evaluate the relative entropy of magic $r_\mathcal{M}$ using the routine for the relative entropy provided in Ref.~\cite{fawzi2018efficient}.

In Fig.~\ref{fig:GUE}a, we study the evolution in time $t$ of random Hamiltonians $H$ drawn from the GUE with $\ket{\psi(t)}=\exp(-iHt)\ket{0}$, starting from initial stabilizer states $\ket{0}$. Here, we normalize $H$ such that its eigenvalues are bounded between $-2$ and $2$~\cite{haug2024probing}. 
We find that the nonstabilizerness increases with time $t$, until converging to a constant for large $t$. 
The increase in nonstabilizerness depends on the chosen magic measure, where we find roughly three different groups. 
BMSA with $\alpha=1/2$ and log-free robustness of magic $\text{LR}$ behave similar, showing relative large values of nonstabilizerness even for short times. 
Next, BMSA with $\alpha=1$ and relative entropy of magic $r_\mathcal{M}$ have nearly the same values, where we find that $\mathcal{A}_1\geq r_\mathcal{M}$ as proven in the main text.
Finally, BMSA with $\alpha=2$ and min-relative entropy of magic $D_\text{min}$ show similar scaling and values, where as expected $\mathcal{A}_2$ is always larger than $D_\text{min}$.

In Fig.~\ref{fig:GUE}b, we study random Clifford circuits doped with $N_\text{T}$ T-gates. We consider a circuit of $N_\text{T}$ layers consisting of randomly sampled Clifford circuits $U_\text{C}^{(k)}$ and the single-qubit T-gate $T=\text{diag}(1,\exp(-i\pi/4)$
\begin{equation}\label{eq:random_CliffT}
\ket{\psi(N_\text{T})}=U_\text{C}^{(0)}\left[\prod_{k=1}^{N_\text{T}} (T\otimes I_{N-1}) U_\text{C}^{(k)} \right]\ket{0}\,.
\end{equation}
We show nonstabilizerness against T-gate density $q=N_\text{T}/N$ for different measures of nonstabilizerness. For all measures, nonstabilizerness increases with $q$ until converging. We find that all our proven bounds are respected. Note that we have $\text{LR}<\mathcal{A}_{1}$ for the Clifford circuits doped with T-gates, while for the random Hamiltonian evolution we find $\text{LR}>\mathcal{A}_{1}$.

\begin{figure*}[htbp]
	\centering	
	\subfigimg[width=0.4\textwidth]{a}{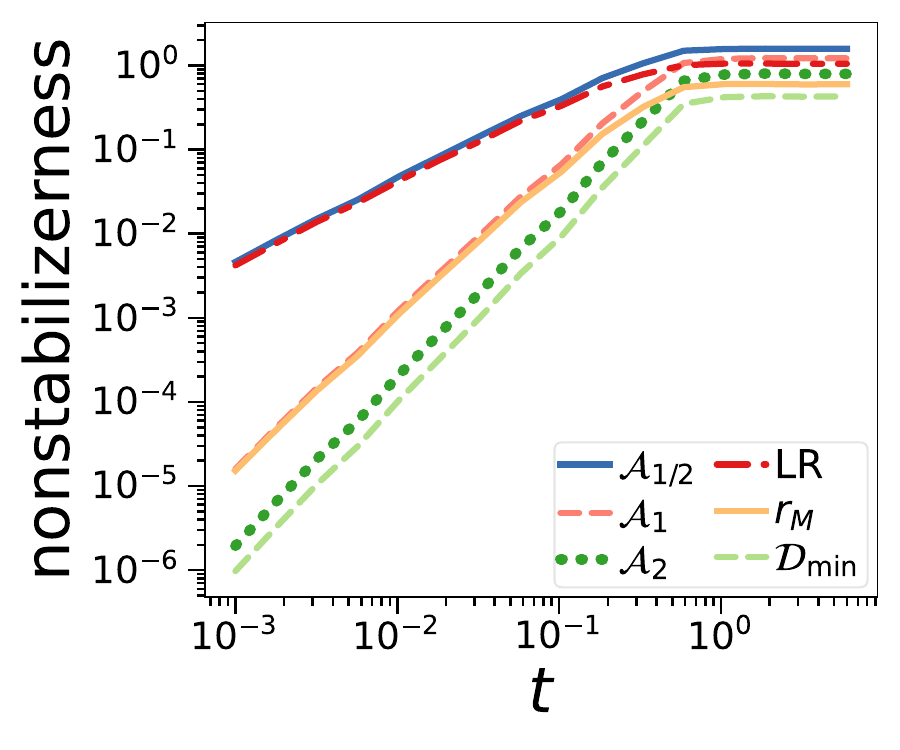}
 \subfigimg[width=0.4\textwidth]{b}{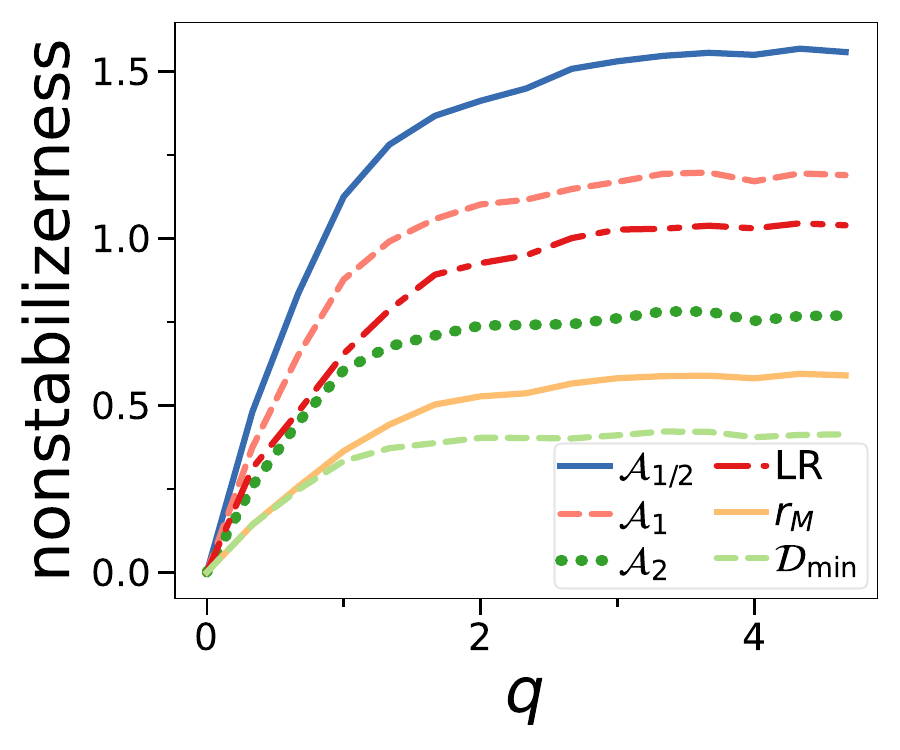}
	\caption{Different measures of nonstabilizerness for \idg{a} random Hamiltonian evolution in time $t$ and \idg{b} doped Clifford circuits with T-gate density $q$. We show BMSA $\mathcal{A}_{\alpha}$ for $\alpha=\{1/2,1,2\}$, log-free robustness of magic $\text{LR}$, relative entropy of magic $r_\mathcal{M}$ and min-relative entropy of magic $D_\text{min}$.  We show $N=3$ qubits where we average over $100$ random instances of the Hamiltonian and circuit.
	}
	\label{fig:GUE}
\end{figure*}

\bibliography{bibliography}
	
\end{document}